\documentclass[runningheads]{llncs}
\usepackage{amsmath, amssymb, bbm}
\usepackage{graphicx}
\usepackage{tikz}
\usetikzlibrary{shapes}
\usetikzlibrary{backgrounds}
\usetikzlibrary{calc}
\usetikzlibrary{positioning}
\newcommand{\bI}{\mathbbm{1}}
\newcommand{\bE}{\mathbb{E}}
\newcommand{\bP}{\mathbb{P}}
\newcommand{\bR}{\mathbb{R}}

\newcommand{\cF}{\mathcal{F}}

\newcommand{\cL}{\mathcal{L}}
\newcommand{\cQ}{\mathcal{Q}}

\newcommand{\vell}{\boldsymbol{\ell}}

\newcommand{\CQC}{\mathsf{CQ}}
\newcommand{\RF}{\mathsf{RF}}

\newcommand{\scaleX}{12}
\newcommand{\scaleY}{1.33}

\newcommand{\updown}{}

\newcommand{\argmin}{\textrm{argmin}}
\newcommand{\argmax}{\textrm{argmax}}

\begin{document}

\title{Revenue-Optimal Deterministic Auctions for Multiple Buyers with Ordinal Preferences over Fixed-price Items
\thanks{The author thanks anonymous reviewers, whose detailed suggestions helped improve and polish the paper.
%This is the full version of a one-page abstract from those proceedings.
}
}
\titlerunning{Assortment Auctions}
\author{
Will Ma\orcidID{0000-0002-2420-4468}
}
\authorrunning{W. Ma}
\institute{Graduate School of Business, Columbia University, New York NY 10027, USA \\ \email{wm2428@gsb.columbia.edu}}
\maketitle              % typeset the header of the contribution
\begin{abstract}
In this paper, we introduce a Bayesian revenue-maximizing mechanism design model where the items have fixed, exogenously-given prices. Buyers are unit-demand and have an ordinal ranking over purchasing either one of these items at its given price, or purchasing nothing. This model arises naturally from the assortment optimization problem, in that the single-buyer optimization problem over deterministic mechanisms reduces to deciding on an assortment of items to ``show''. We study its multi-buyer generalization in the simplest setting of single-winner auctions, or more broadly, any service-constrained environment. Our main result is that if the buyer rankings are drawn independently from Markov Chain ranking models, then the optimal mechanism is computationally tractable, and structurally a virtual welfare maximizer. We also show that for ranking distributions not induced by Markov Chains, the optimal mechanism may not be a virtual welfare maximizer.

\end{abstract}

\section{Introduction}
In this paper,
we study auction design for unit-demand buyers when the prices of the products are fixed.
In particular, a seller is endowed with $n$ substitutable products with exogenously-given prices $r_1,\ldots,r_n\ge0$.
A buyer's outcome and payment from participating in the auction will always take the form ``receive product $j$ and pay price $r_j$'', for some $j=0,\ldots,n$, where we let $j=0$ represent the ``no-purchase'' option with $r_0=0$.

We restrict attention to mechanisms that are \textit{deterministic} and \textit{dominant-strategy} (DS) truthful.
Under this restriction, a buyer's preference is fully captured by a weak ordering of the options ``buy product $j$ for price $r_j$'', for $j=0,\ldots,n$.
Our mechanisms will be DS \textit{incentive-compatible}, where a buyer can never get a more-preferred option from lying about her ranking, and DS \textit{individually-rational}, where a buyer can never get an option she ranks below the no-purchase option.

We study revenue maximization in a Bayesian setting, where the buyers' rankings are drawn from \textit{known}, \textit{independent} (non-identical) distributions.
There is also a feasibility constraint on the inventory of products available to be allocated to buyers.
The problem is then: maximize the expected revenue of a mechanism with respect to these random rankings, subject to it being deterministic, DS incentive-compatible and individually-rational, and satisfying the feasibility constraint in its allocations.

\subsection{Motivation and Related Problems}\label{sec::motivation}

We now explain why we believe this to be a well-motivated problem that arises naturally in relation to the streams of existing literature.

\paragraph{Mechanism design without money.}
Although we use the language of ``prices'' and ``revenue'', we are technically designing mechanisms ``without money'' \cite{schummer2007mechanism}, because the auctioneer cannot charge arbitrary payments.
Our use of ordinal preferences with an outside option has previously appeared in the house allocation problem \cite{shapley1974cores,roth1982incentive,ma1994strategy},
and our focus on deterministic mechanisms and DS truthfulness is inherited\footnote{Although, we should note that there do exist randomized extensions of truthfulness for ordinal preferences \cite{chakrabarty2014welfare}.} from the more general context of allocation mechanisms under one-sided ordinal preferences \cite{svensson1999strategy,bhalgat2011social,sprumont1991division,barbera1997strategy}.
Our paper is different from these lines of work in that we are maximizing the Bayesian expectation of a cardinal objective function, where we have been given real numbers $r_1,\ldots,r_n$ as the rewards for the successful allocations of the products.
By contrast, these lines of work derive settings and conditions under which the feasible space of mechanisms can be nicely characterized, e.g.\ using top-trading cycles \cite{shapley1974cores,roth1982incentive,ma1994strategy}, serial dictatorships \cite{svensson1999strategy}, or the uniform allocation rule \cite{sprumont1991division}.

\paragraph{Bayesian mechanism design.}
On the other hand, the difference between our mechanisms and those for Bayesian unit-demand revenue maximization \cite{chawla2007algorithmic,chawla2015power,chen2015complexity,chen2018complexity}
is that our mechanisms must charge the fixed prices $r_1,\ldots,r_n$, instead of being able to tweak the payments to lure buyers into higher-valued products.
Also, the distributional assumption we will make on the buyers' rankings
(namely, being generated by Markov chains) is combinatorial, and generally incomparable to the assumptions made in this literature (usually on valuations being independent across products).

\paragraph{Assortment optimization and sequential posted assortment.}
When there is a single buyer, our mechanism design problem reduces to the assortment optimization problem, where the mechanism must decide on a subset of products to show\footnote{We will formally prove this reduction using the taxation principle.
Note that the optimal mechanism may not show all products, instead ``hiding'' some lower-priced products to prevent them from being chosen in lieu of higher-priced products.} the buyer.
This is a classical problem in revenue management \cite{talluri2004revenue,farias2013nonparametric,li2015d},
motivated by brick-and-mortar retailers who do not control the pricing but have full control over the ``assortment'' of products carried.
There has also been a recent line of work on the ``sequential posted assortment'' problem \cite{golrezaei2014realtime,ma2020algorithms,rusmevichientong2020dynamic}, motivated by online retailers who recommend personalized subsets to their heterogeneous customers.

Our problem is a generalization of the assortment optimization problem to multiple buyers.
On the other hand, our mechanisms form a superclass of ``sequential posted assortment'' mechanisms,
which is analogous to the well-known relationship between classical auctions and ``sequential posted pricing'' \cite{chawla2010multi}.
In light of these relationships, we will refer to our mechanism design problem as \textbf{assortment auction} problem.

\subsection{Results for Assortment Auctions}\label{sec::introEg}

Our results are focused on the special case of our model where the feasibility constraints are \textit{product-independent}, only depending on the set of ``winners'' allocated a non-zero product but not which specific products they were allocated.
It still captures the fundamental setting of a single-winner auction, as well as the single-leg revenue management problem where $b$ identical units (e.g.\ flight seats, hotel rooms) could be sold using different ``fare classes'' with fixed prices \cite{gallego1994optimal,talluri2004revenue}.
This assumption is generally justified whenever the products correspond to different ``packagings'' of an underlying item, or different ``services'' which share the same limiting resource (e.g.\ different types of massages with the masseur/masseuse, different VIP packages for time backstage).
These settings are also the focus in the papers by \cite{alaei2012bayesian,alaei2013simple}, in which they are called \textit{service-constrained environments}.
We note that our combination of fixed prices with a service-constrained environment is quite natural in the application of the airline selling $b$ seats under different fare classes.

%Formally, our techniques allow for any downward-closed feasible family $\cF$ that the set of winners $W$ must lie in.
%However,
In the Introduction we will describe our results in the further special case of auctions with a \textit{single winner},
%where $\cF=\{W:|W|\le1\}$,
and start to describe our results by relating our problem to classical single-item auctions.
Suppose that the preferences are \textit{buy-down}, where any realizable ranking prefers the non-zero products in order of low-to-high prices.
Such a ranking is characterized by a \textit{valuation}, equal to the maximum price of a product ranked higher than product~0.
In this case, the well-known result of Myerson \cite{myerson1981optimal} says that the optimal mechanism is deterministic, dominant-strategy truthful, and can be implemented in a way where the winner always pays one of the prices in $r_1,\ldots,r_n$.\footnote{
See \cite{elkind2007designing}, who derives Myerson's mechanism for discrete valuations.
Although there do exist randomized, symmetric implementations of Myerson, the one which is deterministic and breaks ties in a consistent order always charges one of the prices in $r_1,\ldots,r_n$, and thus can be implemented as an assortment auction.
%This mechanism is ``remarkable'' in that it is optimal among the larger class of randomized, Bayesian-truthful mechanisms with arbitrary payments.
}
Hence, Myerson's mechanism is the optimal solution in our setting in the special case of buy-down preferences.
Moreover, Myerson's mechanism is structurally a virtual welfare maximizer, in that each buyer, based on only her report and distribution, is assigned a univariate score called a \textit{virtual valuation},
after which the buyer with the highest virtual valuation is declared the winner.
We will call this a ``Myersonian'' structure.

Our first result is a negative one, which shows that for general preference distributions, 
the optimal auction may not have the ``Myersonian'' structure described above.
Namely, when non-zero products are not necessarily ranked in order of low-to-high prices, even with two IID buyers, the optimal allocation rule may want buyer~1 to win the auction when the buyers report the \textit{same} ranking, and buyer~2 to win when they report \textit{different} rankings.
Such an allocation rule clearly cannot be defined via virtual valuations.
Although it was already known that assortment optimization (the special case of our problem with a single buyer) is NP-hard for general preference distributions \cite{aouad2018approximability}, this result shows that even structurally, one cannot hope to derive the optimal auction using only Myersonian mechanisms based on virtual valuations.

Our main positive result is that for preference distributions induced by \textit{Markov Chain choice models}, the optimal mechanism is structurally Myersonian, and computationally tractable.
This is a broad class of choice models introduced in \cite{zhang2005revenue,blanchet2016markov},
where the random ranking satisfies a memorylessness property that the next product in the ranking depends probabilistically on only the current product, and not the entire history.
Markov Chain choice models capture the buy-down preferences corresponding to the classical auctions setting, so our result generalizes Myerson's mechanism for discrete valuations.
They also capture the commonly-used \textit{Multi-Nomial Logit (MNL) choice model}
(a.k.a.\ the Plackett-Luce vase model), as well as the case of single-minded buyers.
Finally, we should mention that the tractability of assortment optimization for Markov Chain choice models was already known \cite{blanchet2016markov,desir2019capacity,feldman2017revenue}, so our auction extends this result to multiple buyers.

\subsection{Description of Optimal Myersonian Auction}

We now explain how our generalized Myersonian mechanism assigns each buyer a virtual valuation based on her reported ranking and ranking distribution.
It is a generalization of how virtual valuations can be assigned in the classical auctions setting based on the ``ironed revenue curve''.
We will consider the following example: there are four products $A,B,C,D$ with prices $r_A=12,r_B=7.5,r_C=4.5,r_D=4$.
One buyer's ranking distribution is uniform over lists $(CBA),(CB),(CD),(C)$, where e.g.\ list $(CB)$ means that her first choice is to buy C at price $r_C$, second choice is to buy B at price $r_B$, and third choice is to buy nothing (we can ignore ordering after the no-purchase option).
%It can be checked that this distribution arises from a Markov chain.

\paragraph{Auction pre-processing.}

Fix a buyer and consider the assortment optimization problem with just that buyer.
For an assortment $S$, let $Q(S)$ denote the probability of selling a product when $S$ is offered, and let $R(S)$ denote the expected revenue.
In the example above, if $S=\{A,B,D\}$, then $Q(S)=\frac{3}{4}$ and $R(S)=\frac{1}{2}r_B+\frac{1}{4}r_D=4.75$.

Now, consider the two-dimensional plot consisting of points $(Q(S),R(S))$ for every assortment $S$.
Call the upper concave envelope of these points the \textit{revenue frontier}.
The revenue frontier is formed by connecting the points for a \textit{sequence of efficient assortments}.
In the example above, this sequence is $\{A\},\{A,D\},\{A,B,D\},\{A,B,C,D\}$ (see Figure~\ref{fig::introEg}).

\begin{figure}[t]
\centering
\begin{tikzpicture}
\draw[->](0,0)--(1*\scaleX,0)node[right]{$Q(S)$};
\draw(0*\scaleX,0)--(0*\scaleX,-0.2)node[below]{\footnotesize 0};
\node[align=center,blue] at (0.14*\scaleX,0.25*\scaleY) {$V(CBA)=12$};
\draw(0.25*\scaleX,0)--(0.25*\scaleX,-0.2)node[below]{\footnotesize 0.25};
\node[align=center,blue] at (0.375*\scaleX,0.25*\scaleY) {$V(CD)=4$};
\draw(0.5*\scaleX,0)--(0.5*\scaleX,-0.2)node[below]{\footnotesize 0.5};
\node[align=center,blue] at (0.625*\scaleX,0.25*\scaleY) {$V(CB)=3$};
\draw(0.75*\scaleX,0)--(0.75*\scaleX,-0.2)node[below]{\footnotesize 0.75};
\node[align=center,blue] at (0.875*\scaleX,0.25*\scaleY) {$V(C)=-1$};
\draw(1*\scaleX,-0.2)--(1*\scaleX,-0.2)node[below]{\footnotesize 1};
\draw[->](0,0)--(0,4.75*\scaleY)node[above]{$R(S)$};

\draw[dashed,gray](0,3*\scaleY)node[left,black]{\footnotesize 3}--(0.25*\scaleX,3*\scaleY);
\draw[dashed,gray](0,4*\scaleY)node[left,black]{\footnotesize 4}--(0.5*\scaleX,4*\scaleY);
\draw[dashed,gray](0,4.75*\scaleY)node[left,black]{\footnotesize 4.75}--(0.75*\scaleX,4.75*\scaleY);
\draw[dashed,gray](0,4.5*\scaleY)node[left,black]{\footnotesize 4.5}--(1*\scaleX,4.5*\scaleY);

\draw[thick,red](0,0)--node[below,sloped,blue]{\footnotesize slope=12}(0.25*\scaleX,3*\scaleY);
\draw[thick,red](0.25*\scaleX,3*\scaleY)--node[below,sloped,blue]{\footnotesize slope=4}(0.5*\scaleX,4*\scaleY);
\draw[thick,red](0.5*\scaleX,4*\scaleY)--node[below,sloped,blue]{\footnotesize slope=3}(0.75*\scaleX,4.75*\scaleY);
\draw[thick,red](0.75*\scaleX,4.75*\scaleY)--node[below,sloped,blue,fill=white]{\footnotesize slope=-1}(1*\scaleX,4.5*\scaleY);

\filldraw (0.25*\scaleX,3*\scaleY)node[above left]{$S=\{A\}$} circle (2pt);
\filldraw (0.5*\scaleX,4*\scaleY)node[above left]{$S=\{A,D\}$} circle (2pt);
\filldraw (0.75*\scaleX,4.75*\scaleY)node[above]{$S=\{A,B,D\}$} circle (2pt);
\filldraw (1*\scaleX,4.5*\scaleY)node[above]{$S=\{A,B,C,D\}$} circle (2pt);
\filldraw (0.25*\scaleX,1*\scaleY)node[above right]{$S=\{D\}$} circle (2pt);
\filldraw (0.5*\scaleX,3.75*\scaleY)node[below right]{$S=\{A,B\}$} circle (2pt);
\end{tikzpicture}
\caption{
The points $(Q(S),R(S))$ plotted for every $S$, all of which are equivalent to one of the six assortments shown.
The revenue frontier is the upper concave envelope in \textcolor{red}{red}, formed by joining the sequence of efficient assortments.
The slopes of the revenue frontier and the resulting virtual valuations V are shown in \textcolor{blue}{blue}.
}
\label{fig::introEg}
\end{figure}
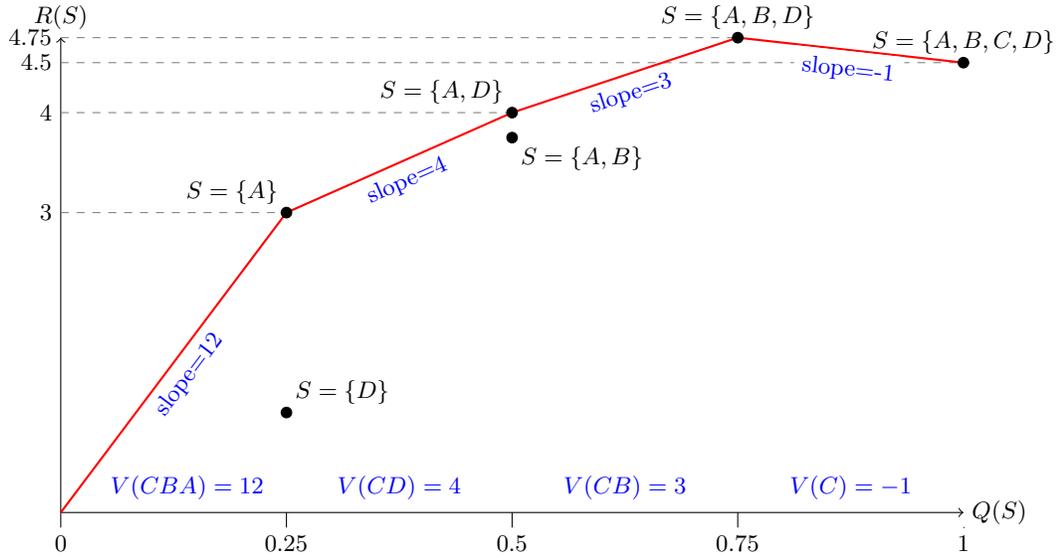

\paragraph{Virtual valuation assignment.}

Fix a buyer and suppose that the sequence of efficient assortments for her list distribution has been identified.
Now, if she reports a realized list of $\ell$, then we find the first efficient assortment $S$ which intersects $\ell$, and set her virtual valuation equal to the slope of the revenue frontier on the \textit{left} side of point $(Q(S),R(S))$.

In the example above, if $\ell=(CD)$, then $\{A,D\}$ is the first efficient assortment to intersect $\ell$, since the earlier assortment in the sequence, $\{A\}$, does not.
The virtual valuation is the slope of the line segment from $(Q(\{A\}),R(\{A\}))$ to $(Q(\{A,D\}),R(\{A,D\}))$, equal to 4 (see Figure~\ref{fig::introEg}).

\paragraph{Winner and allocation.}

Suppose that each buyer has been assigned a virtual valuation based on her list and distribution, as described above.
The winner is then the buyer with the highest virtual valuation.
She is allocated her most-preferred product from the assortment whose left-side slope represents the \textit{minimum} virtual valuation she could have had to win the auction.

In the example above, suppose that the buyer's list realizes to $(CBA)$ and that her virtual valuation of 12 is the highest.
If the second-highest virtual valuation is 3.5, then she would get her most-preferred product from assortment $\{A,D\}$, whose left-side slope is 4 (see Figure~\ref{fig::introEg}).
As a result, she would end up paying 12 for product $A$.
On the other hand, if the second-highest virtual valuation is 2.5, then she would get to choose from assortment $\{A,B,D\}$, whose left-side slope is 3.  As a result, she would get her more-preferred option of paying 7.5 for product $B$.

Like in Myerson's original mechanism, a negative virtual valuation is not allowed to win the auction, and hence even with no competition from other buyers, this buyer would still be restricted to assortment $\{A,B,D\}$ (not $\{A,B,C,D\}$, whose left-side slope is -1).
$\{A,B,D\}$ is our analogue of a \textit{reserve price}, and is the solution to the assortment optimization problem for this single buyer.

We note some novel and surprising aspects of our virtual valuations and mechanism:
\begin{itemize}
\item List $(CB)$ has a lower virtual valuation than list $(CD)$, despite product~$B$ having a higher price than product~$D$.
This is because product~$B$ \textit{cannibalizes} the sale of product~$A$ in the list $(CBA)$, whereas product~$D$ does not.
\item In the example above, the efficient assortments happened to be \textit{nested}, with $\{A\}\subset\{A,D\}\subset\{A,B,D\}$.
This means that less competition (lower virtual valuations from other buyers) leads to more choice for the winner, which is consistent with the intuition from classical auctions.
However, we also provide an assortment auction (\textbf{Example~\ref{eg::nonNested}} in \textbf{Section~\ref{sec::examples}}) for which this is not the case and the classical intuition breaks down.
\end{itemize}

\subsection{Establishing Optimality for Markov Chain Choice Models}

Nothing from the auction described in Section~\ref{sec::introEg} required the list distributions to arise from Markov chains.
We now explain what could go wrong without such an assumption, thereby sketching our proof that the auction is optimal for Markov Chain choice models.

First, we need to show that the revenue frontier and efficient assortments from Section~\ref{sec::introEg} can be tractably computed.
We show that for a Markov Chain choice model, the revenue frontier is always defined by a nested sequence of efficient assortments, which can be greedily constructed (\textbf{Section~\ref{sec::markovChain}}).
Our greedy procedure modifies the ``externality-adjustment'' technique of \cite{desir2019capacity}, and generates a sequence of products which yield the revenue frontier, which is different than the sequence used by \cite{desir2019capacity} for their constrained assortment optimization problem.

Second, we need to show that in the allocation rule from Section~\ref{sec::introEg}, whenever the winner has a positive virtual valuation, the assortment she gets to ``choose from'' actually contains a product she wants.
This is not obvious; in fact, there exist distributions (\textbf{Example~\ref{eg::pathological}} in \textbf{Section~\ref{sec::examples}}) for which the winner may choose no product.
Nonetheless, we introduce a condition called \textit{implementability}, under which the winner always chooses a product
and the \textit{virtual surplus}\footnote{This is the expected value of the maximum among all virtual valuations and 0.} is earned by the Myersonian mechanism (\textbf{Section~\ref{sec::implInsurm}}).
We use the nested property of the efficient assortments of a Markov Chain choice model to establish implementability (\textbf{Section~\ref{sec::mcImplInsurm}}).

Third, we need to show that earning the virtual surplus is optimal.
This does not follow from above; in fact, in our example (\textbf{Example~\ref{eg::nonMyersonian}} in \textbf{Section~\ref{sec::examples}}) where the optimal auction is non-Myersonian, implementability is satisfied and the efficient assortments are even nested!
Nonetheless, we introduce a condition called \textit{insurmountability} which ensures that the virtual surplus cannot be surpassed by a truthful mechanism (\textbf{Section~\ref{sec::implInsurm}}).
We show that it is satisfied under Markov Chain choice models, by proving that the cannibalization patterns from our non-Myersonian example cannot arise under the memorylessness property of Markov chains (\textbf{Section~\ref{sec::mcImplInsurm}}).

We should note that our implementability and insurmountability conditions are similar to the \textit{revenue linearity} condition from \cite{alaei2013simple}, who also studied the existence of Myersonian mechanisms under product-independent feasibility constraints.
However, they focus on cardinal instead of ordinal preferences, and their definitions are stated with respect to Bayesian instead of dominant-strategy truthfulness, which is why we independently develop our conditions using the language of assortment optimization.
Importantly, our main result for Markov Chain ordinal preferences is quite general in that it holds regardless of the prices $r_1,\ldots,r_n$ of the products, which cannot be derived from their results.

\section{Definition of Assortment Auctions} \label{sec::assortmentAuctions}

A seller has a set of products $N=\{1,\ldots,n\}$.
Each product $j\in N$ has a fixed price $r_j$, and the products are indexed so that $0\le r_1\le\cdots\le r_n$.
We also let $j=0$ refer to a ``no-purchase'' product, with $r_0=0$, which is always available.
Throughout this paper, for an arbitrary subset of products $S\subseteq N$, we will let $S_+$ denote the set $S\cup\{0\}$.

There is a set of buyers $M=\{1,\ldots,m\}$.
Each buyer $i\in M$ has a ranked list $\ell_i$ from which she is willing to purchase at most one product, at its corresponding price.
For example, the ranked list $(1,3)$ indicates that the buyer's first choice is to purchase product~1 at price $r_1$, second choice is to purchase product~3 at price $r_3$, and third choice is to make no purchase.
The list omits product~0 and all the products ranked after it, since the ordering after product~0 is irrelevant.
We let $\Omega$ denote the universe of all possible lists, which are the ordered subsets of $N$.

The ranked list can be interpreted as the set of products for which the buyer's utility minus price is positive, sorted in decreasing order.
We treat $\ell_i$ as both a set and a ranking, where $\ell_i(j)$ denotes the rank of product $j$ in list $\ell_i$, with smaller numbers meaning more preferred.
We define $\ell_i(0)=|\ell_i|+1$ and $\ell_i(j)=\infty$ for all $j\in N\setminus\ell_i$.
When presented with an \textit{assortment} of products $S\subseteq N$, a buyer $i$ chooses her most-preferred product from $S_+$, i.e.\ $\argmin_{j\in S_+}\ell_i(j)$.

We let $\vell=(\ell_1,\ldots,\ell_m)$ denote the \textit{list profile}, which consists of all buyers' ranked lists.
We often write $\vell$ as $(\ell_i,\vell_{-i})$, where $\vell_{-i}$ consists of all ranked lists except that of buyer $i$.
A \textit{mechanism} takes in a list profile $\vell$ and outputs an allocation vector $(j_1(\vell),\ldots,j_m(\vell))\in (N_+)^m$, where each buyer $i$ receives product $j_i(\vell)$ and pays $r_{j_i(\vell)}$.

We only consider deterministic mechanisms, because in our setting with ordinal preferences, a buyer's preference under uncertainty is not defined.
Similarly, we only consider dominant-strategy truthfulness, and say that a mechanism is \textit{truthful} if
\begin{align}
\ell_i(j_i(\ell_i,\vell_{-i})) &\le\ell_i(j_i(\ell'_i,\vell_{-i})) &&\forall i\in M,\vell_{-i}\in\Omega^{m-1},\ell_i\in\Omega,\ell'_i\in\Omega; \label{eqn::incentiveCompatible} \\
\ell_i(j_i(\ell_i,\vell_{-i})) &\le\ell_i(0) &&\forall i\in M,\vell_{-i}\in\Omega^{m-1},\ell_i\in\Omega. \label{eqn::individuallyRational}
\end{align}
\eqref{eqn::incentiveCompatible} imposes that the mechanism is \textit{incentive-compatible}, where a buyer $i$ always receives a less-preferred product when she misreports her list as $\ell'_i$ instead of her true list $\ell_i$.
\eqref{eqn::individuallyRational} imposes that the mechanism is \textit{individually-rational}, where a buyer $i$ cannot be forced into purchasing a product which is less-preferred to purchasing nothing.

Hereafter, we restrict to truthful mechanisms and make no distinction between a buyer's reported list and true list.
We also assume that the lists are \textit{strictly ordered}, where we note that the truthfulness constraints~\eqref{eqn::incentiveCompatible}--\eqref{eqn::individuallyRational} are only easier to satisfy if we allow for indifference in the lists.
The following statement provides a useful characterization of all deterministic truthful mechanisms, and is analogous to the ``taxation principle'' for classical auctions.
\begin{proposition}[Taxation Principle for Assortment Auctions] \label{prop::taxation}
Any deterministic truthful mechanism can be characterized by functions $T_i$ for the buyers $i\in M$, where each $T_i$ takes in the other lists $\vell_{-i}$ and outputs an assortment $T_i(\vell_{-i})$, such that the allocation vector satisfies
\begin{align} \label{eqn::taxationConstr}
j_i(\vell)=\argmin_{j\in T_i(\vell_{-i})_+}\ell_i(j) &&\forall i\in M.
\end{align}
\end{proposition}

Proposition~\ref{prop::taxation} can be proven by for each buyer $i$, fixing $\vell_{-i}$ and considering the possible products she could be allocated by the mechanism through the different lies $\ell'_i$ she could tell.  The corresponding assortment $T_i(\vell_{-i})$ she is allowed to choose from in Proposition~\ref{prop::taxation} is then the union $\cup_{\ell'_i\in\Omega}j_i(\ell'_i,\vell_{-i})$.  We provide a detailed proof in Appendix~\ref{apx:secTwo}.

\subsection{Bayesian Revenue Maximization with Product-Independent Feasibility Constraints} \label{sec::bayesianSingleItem}

Our paper assumes the \textit{Bayesian} setting, where the seller has full distributional information about the buyers' private lists.
We assume that the list of a buyer $i$ is drawn \textit{independently} from a distribution over $\Omega$, given by its discrete probability mass function $P_i$, which could be different across buyers.
In the Bayesian setting, we use $\ell_i$ to refer to the realized list of buyer $i$, and $l$ to refer to an arbitrary list in $\Omega$, with $\Pr[\ell_i=l]=P_i(l)$.
We discuss how the distributions over the exponential space $\Omega$ can have compact representations in the form of \textit{choice models}, in Section~\ref{sec::markovChain}.

Our paper also focuses on the case of product-independent feasibility constraint, where the auction is constrained by the set of ``winners'' who are allocated non-zero products, as discussed in the Introduction.  Formally, such a feasibility constraint is described as
\begin{align} \label{eqn::feasibility}
\{i\in M:j_i(\vell)\neq0\}\in\cF &&\forall\vell\in\Omega^m
\end{align}
where $\cF$ is an arbitrary downward-closed feasible family.

We are now ready to define the main problem studied in this paper.
\begin{problem}[Revenue Maximization with Assortment Auctions] \label{prob::aa}
Find the allocation functions $j_1,\ldots,j_m:\Omega^m\to N_+$ which maximize the expected revenue $\bE_{\vell}[\sum_{i=1}^mr_{j_i(\vell)}]$,
subject to truthfulness constraints~\eqref{eqn::incentiveCompatible}-\eqref{eqn::individuallyRational} and feasibility constraint~\eqref{eqn::feasibility}. 
\end{problem}

In the case of a single buyer, Problem~\ref{prob::aa} reduces (via
%the taxation principle in
Proposition~\ref{prop::taxation}) to $\max_{T_1\subseteq N}\bE_{\ell_1}[r_{j_1(\ell_1)}]$, where $j_1(\ell_1)=\argmin_{j\in T_1\cup\{0\}}\ell_1(j)$, which is the basic assortment optimization problem.

We will frequently reference the special case where customers ``buy down'', because it corresponds to known results in the classical auctions setting, which we can then build upon.
In this special case, ranked lists always take the form $(1,\ldots,j)$, for some product $j\in N_+$, whose price $r_j$ corresponds to the customer's maximum willingness-to-pay.
Therefore, we can consider an instance in the classical single-item auction setting where each buyer $i$ has an independent \textit{valuation} $v_i$ that equals $r_j$ with probability $P_i\big((1,\ldots,j)\big)$, for all $j\in N_+$.
Myerson's optimal auction specifies an (ironed) virtual valuation function $\phi_i$ for each buyer $i$, which maps the discrete valuation set $\{r_0,\ldots,r_n\}$ to $\bR$ \cite{elkind2007designing}.
This can then be translated back into an optimal assortment auction, in the special case of buy-down preferences (see Appendix~\ref{apx:secTwo} for details).
\begin{proposition} \label{prop::myePricing}
Suppose that the list distribution $P_i$ for every buyer $i$ is supported within $\{(1,\ldots,j):j\in N_+\}$ and that $\cF=\{M'\subseteq M:|M'|\le1\}$.
Then the optimal auction is:
\begin{itemize}
\item Each buyer $i$, upon reporting list $(1,\ldots,j)$, is assigned a virtual valuation of $\phi_i(r_j)$;
\item Declare the winner to be the buyer $i^*$ with the highest\footnote{With discrete valuations, we can perturb the functions $\phi_1,\ldots,\phi_m$ slightly so that different buyers cannot have the same virtual valuation.  This is equivalent to using an arbitrary deterministic tie-breaking rule.} positive\footnote{Only a positive virtual valuation can win the auction; otherwise no buyer is allocated any product.} virtual valuation, and allocate her the lowest-priced product $j$ whose virtual valuation $\phi_{i^*}(r_j)$ would have won the auction.
\end{itemize}
Moreover, the optimal revenue equals the expected virtual surplus, defined as the
expected value of the maximum virtual valuation (or zero if all virtual valuations are non-positive).
\end{proposition}
%Proposition~\ref{prop::myePricing} can be easily verified by checking that an assortment auction under buy-down preferences reduces to a classical auction.
%Proposition~\ref{prop::myePricing} will also be a special case of our general results in Section~\ref{sec::markovChain} for assortment auctions under Markov Chain choice models.

\subsection{Implementable and Insurmountable Virtual Valuations} \label{sec::implInsurm}

One way of interpreting Myerson's result for classical auctions is that given any valuation distributions, it is \textit{always} possible to find functions $\phi_1,\ldots,\phi_n$ which are simultaneously:
\begin{itemize}
\item ``Low enough'', in that the expected virtual surplus (defined according to $\phi_1,\ldots,\phi_n$) can be earned by a mechanism which allocates to the highest virtual valuation;
\item ``High enough'', in that the expected virtual surplus is an upper bound on the revenue of any feasible truthful mechanism.
\end{itemize}
We will see in Section~\ref{sec::examples} that for general preference distributions, the two conditions above cannot always be simultaneously satisfied.
Nonetheless, in this section we derive what it means for virtual valuations to be ``low enough'' and ``high enough'' for assortment auctions.

For assortment auctions, the natural generalization of a function $\phi$ which maps valuations to virtual valuations is a function $V$ defined on lists.
\begin{definition}
A \textit{virtual valuation mapping (VVM)} is a function $V:\Omega\to\bR$.  Given a VVM $V_i$ for every buyer $i\in M$ and a downward-closed feasible family $\cF$, define the \textit{expected virtual surplus} as
\begin{align} \label{eqn::evs}
\bE_{\vell}\left[\max_{M'\in\cF}\sum_{i\in M'}V_i(\ell_i)\right].
\end{align}
\end{definition}
For convenience, a VVM $V_i$ will often leave virtual valuations undefined for lists of measure zero.
The virtual valuation of the empty list $\emptyset$ is understood to be $-\infty$, which can always be excluded in the accounting of the virtual surplus since $\cF$ is downward-closed.

We now introduce conditions on individual VVM's $V$, omitting the buyer subscript $i$.
In the classical auctions setting, $\phi$ was a function defined based on a valuation distribution.
Similarly, in our setting, a VVM $V$ is always defined based on a specific list distribution $P$.

\begin{definition} \label{def::RQ}
Fix a list distribution $P$.
For all assortments $S\subseteq N$, define the following:
\begin{itemize}
\item $\cQ^j(S)=\{l\in\Omega:j=\argmin_{j\in S_+}l(j)\}$, the subset of lists $l$ which choose product $j$ when offered assortment $S$, defined for all $j\in S$;
\item $\cQ(S)=\bigcup_{j\in S}\cQ^j(S)$, the subset of lists $l$ which make a purchase when offered $S$;
\item $R(S)=\sum_{j\in S}r_j\cdot P(\cQ^j(S))$, the expected revenue when offering $S$ to distribution $P$;
\item $Q(S)=P(\cQ(S))$, the probability of getting a sale when offering $S$ to distribution $P$.
\end{itemize}
\end{definition}

We now introduce the first of our two conditions on a VVM $V$ for a distribution $P$.
\begin{definition} \label{def::impl}
We say that a virtual valuation mapping $V$ for a distribution $P$ is \textit{implementable} if for any threshold $w\in\bR$, we can find an assortment $S$ such that $\cQ(S)=\{l\in\Omega:V(l)\ge w\}$ and
\begin{align} \label{eqn::impl}
\sum_{l\in\cQ(S)}V(l)P(l)\le R(S).
\end{align}
\end{definition}
It can be checked that the virtual valuations defined for the example in Section~\ref{sec::introEg} satisfy this condition.
Indeed, the relevant thresholds are $w=12,4,3,-1$, and for each of these thresholds we can find the respective assortments $S=\{A\},\{A,D\},\{A,B,D\},\{A,B,C,D\}$ which satisfy~\eqref{eqn::impl} as equality.
Note that in Definition~\ref{def::impl}, it is important for $V(\emptyset)$ to be understood to be $-\infty$, since the empty list $\emptyset$ cannot lie in $\cQ(S)$ for any assortment $S$.

If for every buyer, the VVM defined for her is implementable, then the virtual valuations are ``low enough'' in the sense described earlier.
\begin{lemma} \label{lem::LB}
Suppose that VVM $V_i$ is implementable for distribution $P_i$ for all $i\in M$.
Then revenue equal to the expected virtual surplus~\eqref{eqn::evs} can be attained by a Myersonian mechanism, which on each realization of $\vell$ offers assortments in a way so that buyers in $\argmax_{M'\in\cF}\sum_{i\in M'}V_i(\ell_i)$ are allocated a non-zero product.
\end{lemma}

The proof of Lemma~\ref{lem::LB} is deferred to Appendix~\ref{apx:secTwo}.  The virtual valuations for the example from Section~\ref{sec::introEg} were defined based on the efficient assortments and the revenue frontier.
If these virtual valuations satisfy implementability (which is not always the case---see Example~\ref{eg::pathological} in Section~\ref{sec::examples}), then they will maximally inflate the expected virtual surplus earned by the Myersonian mechanism.
One may hope that this means the Myersonian mechanism is optimal.
Surprisingly though, it is not sufficient for the optimality of the Myersonian mechanism, which motivates the need for our second condition.
\begin{definition} \label{def::insurm}
We say that a virtual valuation mapping $V$ for distribution $P$ is \textit{insurmountable} if for all assortments $S\subseteq N$,
\begin{align} \label{eqn::insurm}
\sum_{l\in\cQ(S)}V(l)P(l)\ge R(S).
\end{align}
\end{definition}
The virtual valuations defined for the example in Section~\ref{sec::introEg} are also insurmountable, although this is much more difficult to check.
In the special case of buy-down preferences, $\cQ(S)$ always corresponds to a contiguous block of lists exceeding some price threshold, and insurmountability becomes a trivial condition, which is always satisfied by Myerson's ironed virtual valuations.
However, with general preference lists, there are exponentially many possibilities for $\cQ(S)$, and whether insurmountability holds depends on the specific substitution patterns across those preference lists.
In fact, an assortment $S$ can violate~\eqref{eqn::insurm} even when it is ``inefficient'' and has a small value of $R(S)$, especially if $\cQ(S)$ contains many lists with low virtual valuations.
And when~\eqref{eqn::insurm} is violated, a non-Myersonian mechanism can indeed surpass the expected virtual surplus.

Nonetheless, if for every buyer, the VVM defined for her is insurmountable, then the expected virtual surplus cannot be surpassed.
\begin{lemma} \label{lem::UB}
Suppose that VVM $V_i$ is insurmountable for distribution $P_i$ for all $i\in M$.  Then the expected virtual surplus~\eqref{eqn::evs} is an upper bound on the revenue of any feasible truthful mechanism.
\end{lemma}

The proof of Lemma~\ref{lem::UB} is deferred to Appendix~\ref{apx:secTwo}.  Combining Lemmas~\ref{lem::LB},~\ref{lem::UB}, we see that if for each buyer we can find a VVM which is simultaneously implementable and insurmountable, then a Myersonian mechanism is optimal.

\subsection{Instructive Examples} \label{sec::examples}

We now see an example where the optimal mechanism is non-Myersonian.  It is not possible to find VVM's which are simultaneously implementable and insurmountable.
\begin{example}\label{eg::nonMyersonian}
There are four products $A,B,C,D$ with prices $r_A=4,r_B=2,r_C=r_D=1$.
There are two IID buyers, whose ranking distributions are uniform over lists $(BA),(CBD),(B),(C)$, participating in a single-winner auction.
For this distribution, it can be checked that the revenue frontier is formed by \textbf{nested} assortments $\{A\},\{A,B\},\{A,B,C\}$, which have respective sales probabilities $0.25,0.75,1$ and expected revenues $1,1.5,1.5$, and that the VVM $V$ based on it, defined by $V(BA)=4,V(CBD)=V(B)=1,V(C)=0$, is implementable.
Despite the revenue frontier being formed by nested assortments, $V$ is not insurmountable, because
\begin{align*}
\sum_{l\in\cQ(\{A,C\})}V(l)P(l)=1.25<1.5=R(\{A,C\})
\end{align*}
and hence inequality~\eqref{eqn::insurm} is violated for assortment $\{A,C\}$.
Nonetheless, we can still define virtual valuations for both buyers using $V$, and earn expected virtual surplus 
$
\frac{7\cdot4+8\cdot1}{16}=\frac{36}{16}
$
using a Myersonian mechanism, as shown in Lemma~\ref{lem::LB}.

However, now consider the mechanism specified in Figure~\ref{fig::nonMyersonian}.  It is feasible, truthful, and earns
$
\frac{5r_A+6r_B+4r_C+r_D}{16}=\frac{37}{16},
$
which is greater than $\frac{36}{16}$.
It is easy to see that this mechanism cannot be implemented using virtual valuations.
Indeed, in Figure~\ref{fig::nonMyersonian}, conditioned on both buyers' lists being (B) or (C), buyer~2 wins if and only if buyer 1 reports the \textit{same} list.
By contrast, in a Myersonian mechanism which maximizes virtual valuation, whether buyer~2 wins must be determined by comparing a single number.
\begin{figure}[t]
\centering
\begin{tabular}{c|c||c|c|c|c}
\updown & $\ell_2$ & $(BA)$ & $(CBD)$ & $(B)$ & $(C)$ \\
\hline
\updown $\ell_1$ & $T_i(\ell_{-i})$ & $\{A\}$ & $\{A\}$ & $\{A,C\}$ & $\{A,B\}$ \\
\hline\hline
\updown $(BA)$ & $\emptyset$ & \textbf{A} & \textbf{A} & \textbf{A} & \textbf{B} \\
\hline
\updown $(CBD)$ & $\{A,D\}$ & $A$ & $D$ & \textbf{C} & \textbf{B} \\
\hline
\updown $(B)$ & $\{A,B\}$ & $B$ & $B$ & $B$ & \textbf{B} \\
\hline
\updown $(C)$ & $\{A,C\}$ & $A$ & $C$ & \textbf{C} & $C$ \\
\end{tabular}
\caption{
The optimal non-Myersonian mechanism for this example, which specifies the product allocated for each possible combination of the buyers' lists $\ell_1,\ell_2$.
Bolded entries mean that buyer~1, the row player, wins; non-bolded entries mean that buyer~2 wins.
To see that the mechanism is truthful, the ``taxation'' assortments $T_i(\ell_{-i})$ are shown; e.g.\ $\{A,D\}$ means that when $\ell_1=(CBD)$, buyer~2 gets to choose from assortment $\{A,D\}$.
}
\label{fig::nonMyersonian}
\end{figure}
\end{example}

We also provide an example where virtual valuations based on the revenue frontier, as illustrated in Section~\ref{sec::introEg}, are not implementable.
This leads to a mechanism which may not allocate any product to the buyer with the highest positive virtual valuation, which means that the virtual surplus cannot be earned.
\begin{example}\label{eg::pathological}
There are three products $A,B,C$ with prices $r_A=25,r_B=15,r_C=12$.
The ranking distribution is uniform over lists $(BA),(CB),(B),(B),(C)$.
It can be checked (see calculations in Appendix~\ref{apx:secTwo}) that the revenue frontier is formed by sets $\emptyset,\{A\},\{A,C\},\{A,B\},\{A,B,C\}$ and that there is no way to define corresponding virtual valuations which satisfy Definition~\ref{def::impl}.
\end{example}

Finally, we provide an example where virtual valuations based on \textbf{non-nested} efficient assortments are both implementable and insurmountable.
In this case, even though there is an optimal Myersonian mechanism, less competition from other buyers (in the form of lower virtual valuations) does not necessarily mean that the assortment available for the winner to choose from (as described in Lemma~\ref{lem::LB}) is set-wise increasing.
As a result, less competition could actually result in the winner receiving a less-preferred outcome.
\begin{example}\label{eg::nonNested}
There are three products $A,B,C$ with prices $r_A=6,r_B=3,r_C=2$.
The ranking distribution is uniform over lists $(BA),(CB),(B)$.
It can be checked (see calculations in Appendix~\ref{apx:secTwo}) that the virtual valuations corresponding to the revenue frontier formed by non-nested sets $\emptyset,\{A\},\{A,C\},\{B\}$ are both implementable and insurmountable.
\end{example}

\section{Optimal Assortment Auction for Markov Chains} \label{sec::markovChain}

In this section we derive the optimal assortment auction under Markov Chain choice models, by specifying a procedure for defining a buyer's virtual valuation based on her reported list and list distribution.
Our procedure in essence efficiently constructs the revenue frontier and defines virtual valuations following the example in Section~\ref{sec::introEg}.
We focus on a single buyer and omit the subscript $i$ in this section.

The Markov Chain choice model was introduced in \cite{zhang2005revenue,blanchet2016markov}.
We provide an equivalent definition below, which generates a random ranked list of products according to a Markov chain.

\begin{definition}[Markov Chain Choice Model]
Under a Markov Chain choice model, the list distribution $P$ is implicitly defined in the following way.
There is a Markov Chain with node set $N_+$ (recall that for any set $S$, we defined $S_+=S\cup\{0\}$).
For all nodes $j\in N$ and $j'\in N_+$, the probability of transitioning from node $j$ to node $j'$ is $\rho_{jj'}$.
The outgoing probabilities from every node $j\in N$ satisfy $\sum_{j'\in N_+}\rho_{jj'}=1$, and 0 is a \textit{terminal} node with no outgoing transitions.

To generate a list $\ell\in\Omega$ according to distribution $P$, we start at each node $j'\in N_+$ with probability $\lambda_{j'}$ (these probabilities satisfy $\sum_{j'\in N_+}\lambda_{j'}=1$), in which case we start with the singleton list $(j')$.
We then transition probabilistically along the Markov chain, adding every node visited to the end of the list, but \textit{only if} that node doesn't already appear on the list.
The list immediately ends upon terminal node~0 being reached (and~0 is never added to the list).
It is assumed that 0 is the only absorbing state in the Markov chain, so that $\ell$ terminates with probability 1.
\end{definition}

\begin{definition}[Notation]
In this section, we will use the following notation, which facilitates the analysis of the Markov Chain choice model.
\begin{itemize}
\item For any $j\in N_+$, we will often use $j$ to refer to the singleton set $\{j\}$.
\item Consider the discrete probability space defined by the Markov Chain's distribution $P$.
For any $s\in N_+$ and $S\subseteq N_+\setminus s$, let $s\prec S$ denote the event that node $s$ is visited before any of the nodes in $S$, and let $\bP[s\prec S]$ denote its probability.
\item Similarly, for any $s\in N$, $s'\in N_+\setminus s$, and $S\subseteq N_+\setminus\{s,s'\}$, let $\bP[s\prec s'\prec S]$ denote the probability that $s$ is visited before $s'$, which in turn is visited before any of the nodes in $S$.
Let $\bP_{s}[s'\prec S]$ denote the probability that \textit{starting from node $s$} (instead of starting according to the probabilities $\lambda_k$), $s'$ is visited before any of the nodes in $S$.
%Note that
%\begin{align*}
%\bP[s\prec s'\prec S]=\bP[s\prec s'\cup S]\bP_{s}[s'\prec S],
%\end{align*}
%by the memorylessness of the Markov chain.
\item For $s\in N$, $s'\in N_+\setminus s$, and $S\subseteq N_+\setminus\{s,s'\}$, let $\cL(s\prec s'\prec S)$ denote the subset of lists $l\in\Omega$ for which $s\prec s'\prec S$.  Note that although the list $l$ is truncated upon reaching node 0, since $s\neq0$, whether $s\prec s'\prec S$ is fully determined by $l$.
\end{itemize}
\end{definition}

We are now ready to define our virtual valuation mapping $V:\Omega\to\bR$.
Based on the distribution $P$, our procedure constructs a sequence of products $s^{(1)},\ldots,s^{(K)}$, where each product $s^{(k)}$ maximizes the \textit{incremental efficiency} when added to assortment $S^{(k-1)}=\{s^{(1)},\ldots,s^{(k-1)}\}$.
These efficiencies are computed using \textit{externality-adjusted prices}, which were introduced in \cite{desir2019capacity}.
We emphasize, however, that our procedure is different from theirs, in that our sequence of products forms the revenue frontier while theirs accomplishes a different purpose, as illustrated in Example~\ref{eg::procedure} below.
Having constructed the revenue frontier, we define the virtual valuation of a list $l$ to be the incremental efficiency of the \textit{first} product in our sequence $s^{(1)},\ldots,s^{(K)}$ which appears in $l$.

\begin{definition}[Procedure for Defining Virtual Valuations]\label{def::procedure}\
\begin{itemize}
\item Initialize assortment $S^{(0)}=\emptyset$, and $r^{(0)}_j=r_j$ for all products $j\in N$.
\item For iterations $k=1,2,\ldots$
\begin{enumerate}
\item Set $s^{(k)}$ to be a product $j\in N\setminus S^{(k-1)}$ with $\bP_j[0\prec S^{(k-1)}]\neq0$ which maximizes the \textit{incremental efficiency}, defined as
\begin{align}
\frac{r^{(k-1)}_j}{\bP_j[0\prec S^{(k-1)}]}. \label{eqn::incrEff}
\end{align}
If there are no such products satisfying $\bP_j[0\prec S^{(k-1)}]\neq0$, then STOP.
\item Define the virtual valuation $V(l)$ for every list $l\in\cL(s^{(k)}\prec0\prec S^{(k-1)})$ to be this maximum incremental efficiency, equal to expression~\eqref{eqn::incrEff} with $j=s^{(k)}$. \label{step::setVV}
\item Update the assortment after iteration $k$ to be $S^{(k)}=S^{(k-1)}\cup\{s^{(k)}\}$, and update the \textit{externality-adjusted prices} for all remaining products $j\in N\setminus S^{(k)}$ to be
$$r^{(k)}_j=r^{(k-1)}_j-r^{(k-1)}_{s^{(k)}}\bP_j[s^{(k)}\prec S^{(k-1)}_+].$$
\end{enumerate}
\end{itemize}
After the procedure stops, define $K$ to be last iteration on which virtual valuations were defined, i.e.\ $K=|S^{(k)}|$, which is at most $n$.
\end{definition}

\begin{remark}[Validity of Procedure] \label{rem::procedure}
Note that on each iteration $k$ of the procedure:
\begin{itemize}
\item The lists with virtual valuations already defined contain some product in $\{s^{(1)},\ldots,s^{(k-1)}\}=~S^{(k-1)}$;
\item The new lists for which virtual valuations are defined is $\cL(s^{(k)}\prec0\prec S^{(k-1)})$, which do not contain any product in $S^{(k-1)}$, and hence do not conflict with existing definitions;
\item When the procedure stops, all lists containing some product in $S^{(K)}$ have been covered.
\end{itemize}
Any non-empty lists whose virtual valuations remain undefined at the end must have $j\prec0\prec S^{(K)}$ for some product $j\in N\setminus S^{(K)}$.
By the stopping criterion, $\bP_j[0\prec S^{(K)}]=0$ for any such $j$, which implies that $\bP[j\prec0\prec S^{(K)}]=0$.
Therefore, the non-empty lists $l\in\Omega$ with $V(l)$ undefined have measure zero, and are inconsequential.
\end{remark}

We now make several further remarks.
Note that the procedure from Definition~\ref{def::procedure} does not require knowing the values of $\lambda_{j'}$, because it is only based on probabilities which start at a \textit{specific} node $j$ in the Markov chain.
Also, these probabilities can be tractably computed using the methodology described in \cite{blanchet2016markov}.

Nonetheless, if we fix any values of $\lambda_0,\ldots,\lambda_n$, then~\eqref{eqn::incrEff} can indeed be interpreted as the incremental efficiency.
Take an iteration $k$, where we are considering adding a product $j\in N\setminus S^{(k-1)}$ to assortment $S^{(k-1)}$.
The increase in sales probability, $Q(S^{(k-1)}\cup j)-Q(S^{(k-1)})$, is the measure of the lists in $\cL(j\prec0\prec S^{(k-1)})$, which make a purchase with product $j$ but did not make a purchase without.
The increase in revenue, $R(S^{(k-1)}\cup j)-R(S^{(k-1)})$, is the externality-adjusted price, $r^{(k-1)}_j$, times the probability that $j$ is chosen from assortment $S^{(k-1)}\cup j$, which is $\bP[j\prec S^{(k-1)}_+]$.  (We will formally show this in Proposition~\ref{prop::mcInduction}, but this fact comes from \cite{desir2019capacity}.)
Therefore,
\begin{align}
\frac{R(S^{(k-1)}\cup j)-R(S^{(k-1)})}{Q(S^{(k-1)}\cup j)-Q(S^{(k-1)})}=\frac{r^{(k-1)}_j\bP[j\prec S^{(k-1)}_+]}{\bP[j\prec0\prec S^{(k-1)}]} \label{eqn::vvEqualsIncrEff}
%\nonumber \\
%&=\frac{r^{(k-1)}_j\bP[j\prec S^{(k-1)}_+]}{\bP[j\prec S^{(k-1)}_+]\bP_j[0\prec S^{(k-1)}]}, 
\end{align}
and hence the incremental efficiency is equivalent to expression~\eqref{eqn::incrEff}.

Our virtual valuation generation procedure is different from the iterative assortment optimization procedure of \cite{desir2019capacity} in that at each iteration, it maximizes the incremental efficiency
%$\frac{R(S^{(k-1)}\cup j)-R(S^{(k-1)})}{Q(S^{(k-1)}\cup j)-Q(S^{(k-1)})}=\frac{r^{(k-1)}_j}{\bP_j[0\prec S^{(k-1)}]}$,
instead of the externality-adjusted price $r^{(k-1)}_j$.
Their procedure addresses the single-buyer problem under additional capacity and knapsack constraints.
By contrast, our procedure addresses the unconstrained problem for multiple buyers, where the increase in sales probability (corresponding to the term in our denominator) is highly relevant, as it prevents allocations from being made to another buyer.
The following example demonstrates our procedure and the difference from their procedure.

\begin{example}[Illustration of our Procedure and its difference from \cite{desir2019capacity}] \label{eg::procedure}
There are four products $A,B,C,D$ with prices $r_A=6,r_B=5,r_C=4,r_D=3$.
The ranking distribution is uniform over lists $(CBA),(CB),(CD),(D)$.
It is easy to see that this distribution is generated by the Markov chain pictured in Figure~\ref{fig::procedure}.
Iterations $k=1,2,3$ of the procedure are depicted in Figure~\ref{fig::procedure}, with the externality-adjusted prices $r^{(k-1)}_j$ shown on the nodes.

In iteration $k=1$, the highest-priced product~$A$ is selected, with a virtual valuation of $r_A=6$.

In iteration $k=2$, our procedure diverges from what \cite{desir2019capacity} would do, in that it selects product~$B$, instead of product~$D$ which has the highest externality-adjusted price.
This is because we care about minimizing the denominator of~\eqref{eqn::incrEff}, in which $\bP_B[0\prec S^{(1)}]=1/2$ while $\bP_D[0\prec S^{(1)}]=1$, reflecting the fact that we want to minimize the increase in sales probability.
As a result, even though $r^{(1)}_B<r^{(1)}_D$ and $R(S^{(1)}\cup B)<R(S^{(1)}\cup D)$, our procedure adds product~$B$ to assortment $S^{(1)}=\{A\}$, which maximizes the gain in revenue relative to the increase in sales probability.

In iteration $k=3$, our procedure selects product~$D$ and then stops.  Even though the assortment $S^{(3)}=\{A,B,D\}$ does not contain~$C$, the stopping criterion $\bP_C[0\prec S^{(3)}]=0$ is met.
It can also be checked that $r^{(3)}_C=r^{(2)}_C-r^{(2)}_D\bP_C[D\prec S^{(2)}]=-1/3$, which is negative.
\end{example}

\begin{figure}[t]
\centering
\scalebox{.9}{
\begin{tikzpicture}
\node[draw, ellipse, minimum width={75pt}] (A) {$\mathbf{r^{(0)}_A=6}$};
\node at ($(A)-(2,0)$) {$k=1$:};
\node[right = of A, draw, ellipse, minimum width={75pt}] (B) {$r^{(0)}_B=5$};
\node[right = of B, draw, ellipse, minimum width={75pt}] (C) {$r^{(0)}_C=4$};
\node[right = of C, draw, ellipse, minimum width={75pt}] (D) {$r^{(0)}_D=3$};
\node[draw, circle] (0) at ($(B)!0.5!(C)-(0,1.5)$) {0};
\draw[->](C)--node[midway, fill=white]{\footnotesize $\frac{2}{3}$}(B);
\draw[->](C)--node[midway, fill=white]{\footnotesize $\frac{1}{3}$}(D);
\draw[->](B)--node[midway, fill=white]{\footnotesize $\frac{1}{2}$}(A);
\draw[->](B)--node[midway, fill=white]{\footnotesize $\frac{1}{2}$}(0);
\draw[->](A)--(0);
\draw[->](D)--(0);
\draw[->]($(C)+(0,0.75)$)node[above]{\footnotesize $\lambda=\frac{3}{4}$}--(C);
\draw[->]($(D)+(0,0.75)$)node[above]{\footnotesize $\lambda=\frac{1}{4}$}--(D);
\node[align=left] (V) at ($(D)-(0,1.5)$) {$\frac{r^{(0)}_A}{\bP_A[0\prec S^{(0)}]}=\frac{6}{1}=6$, $\mathbf{s^{(1)}=A}$};
\end{tikzpicture}
}

%\

\scalebox{.9}{
\begin{tikzpicture}
\node[draw, ellipse, minimum width={75pt}] (A) {$A$};
\node at ($(A)-(2,0)$) {$k=2$:};
\node[right = of A, draw, ellipse, minimum width={75pt}] (B) {$r^{(1)}_B=2$};
\node[right = of B, draw, ellipse, minimum width={75pt}] (C) {$r^{(1)}_C=2$};
\node[right = of C, draw, ellipse, minimum width={75pt}] (D) {$r^{(1)}_D=3$};
\node[draw, circle] (0) at ($(B)!0.5!(C)-(0,1.5)$) {0};
\draw[->](C)--node[midway, fill=white]{\footnotesize $\frac{2}{3}$}(B);
\draw[->](C)--node[midway, fill=white]{\footnotesize $\frac{1}{3}$}(D);
\draw[->](B)--node[midway, fill=white]{\footnotesize $\frac{1}{2}$}(A);
\draw[->](B)--node[midway, fill=white]{\footnotesize $\frac{1}{2}$}(0);
\draw[->](A)--(0);
\draw[->](D)--(0);
\draw[->]($(C)+(0,0.75)$)node[above]{\footnotesize $\lambda=\frac{3}{4}$}--(C);
\draw[->]($(D)+(0,0.75)$)node[above]{\footnotesize $\lambda=\frac{1}{4}$}--(D);
\node[align=left] (V) at ($(D)-(0,1.5)$) {$\frac{r^{(1)}_B}{\bP_B[0\prec S^{(1)}]}=\frac{2}{1/2}=4$, $\mathbf{s^{(2)}=B}$};
\end{tikzpicture}
}
%\

\scalebox{.9}{
\begin{tikzpicture}
\node[draw, ellipse, minimum width={75pt}] (A) {$A$};
\node at ($(A)-(2,0)$) {$k=3$:};
\node[right = of A, draw, ellipse, minimum width={75pt}] (B) {$B$};
\node[right = of B, draw, ellipse, minimum width={75pt}] (C) {$r^{(2)}_C=\frac{2}{3}$};
\node[right = of C, draw, ellipse, minimum width={75pt}] (D) {$r^{(2)}_D=3$};
\node[draw, circle] (0) at ($(B)!0.5!(C)-(0,1.5)$) {0};
\draw[->](C)--node[midway, fill=white]{\footnotesize $\frac{2}{3}$}(B);
\draw[->](C)--node[midway, fill=white]{\footnotesize $\frac{1}{3}$}(D);
\draw[->](B)--node[midway, fill=white]{\footnotesize $\frac{1}{2}$}(A);
\draw[->](B)--node[midway, fill=white]{\footnotesize $\frac{1}{2}$}(0);
\draw[->](A)--(0);
\draw[->](D)--(0);
\draw[->]($(C)+(0,0.75)$)node[above]{\footnotesize $\lambda=\frac{3}{4}$}--(C);
\draw[->]($(D)+(0,0.75)$)node[above]{\footnotesize $\lambda=\frac{1}{4}$}--(D);
\node[align=left] (V) at ($(D)-(0,1.5)$) {$\frac{r^{(2)}_D}{\bP_D[0\prec S^{(2)}]}=\frac{3}{1}=3$, $\mathbf{s^{(2)}=D}$};
\end{tikzpicture}
}

\caption{
The demonstration of our procedure from Definition~\ref{def::procedure} on Example~\ref{eg::procedure}.
In the Markov chain, any arrival probabilities $\lambda$ or transition probabilities $\rho$ not depicted are 0.
In each iteration $k$, the highest incremental efficiency (given by~\eqref{eqn::incrEff}) and the product selected $s^{(k)}$ are shown.
}
\label{fig::procedure}
\end{figure}
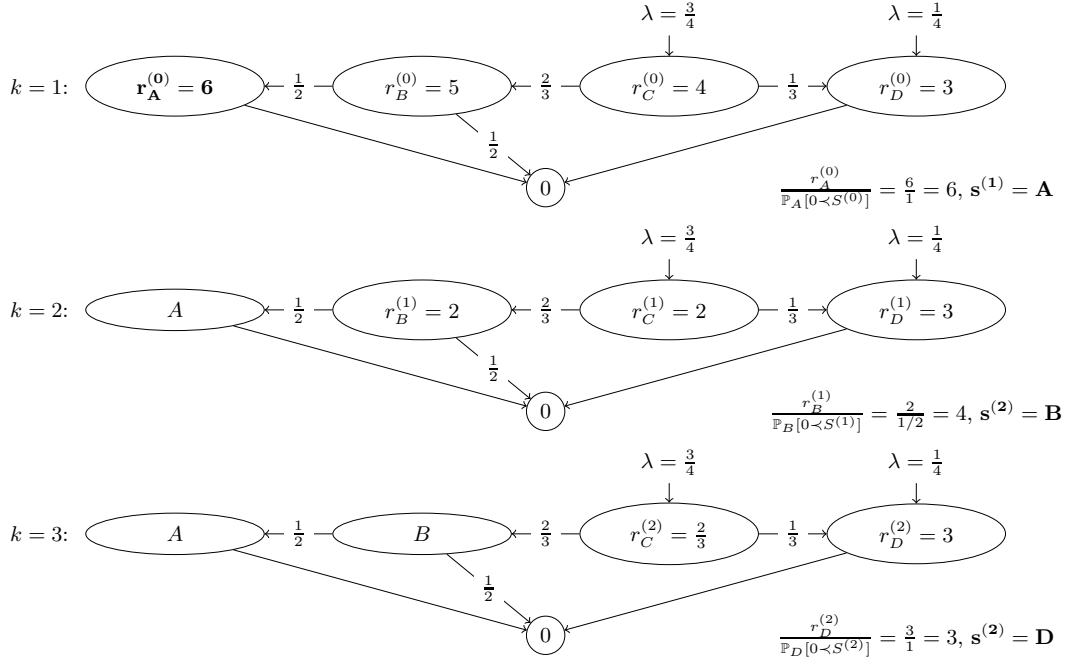

\subsection{Implementability and Insurmountability of Procedure} \label{sec::mcImplInsurm}

We now build toward showing that the virtual valuation mapping $V:\Omega\to\bR$ constructed in Definition~\ref{def::procedure} is both implementable and insurmountable.

\begin{definition} \label{def::kappaV}
Consider the procedure from Definition~\ref{def::procedure}.
\begin{enumerate}
\item For a
list
$\ell$, let $\kappa(\ell)$ denote the index of the first product in sequence $s^{(1)},\ldots,s^{(K)}$ to appear in $\ell$, which is well-defined (see Remark~\ref{rem::procedure}) as long as $\ell\neq\emptyset$ and has positive measure.
\item For all $k=1,\ldots,K$, let $V^{(k)}$ denote the virtual valuation of lists $\ell$ with $\kappa(\ell)=k$, equal to expression~\eqref{eqn::incrEff} with $j=s^{(k)}$.
Note that if $Q(S^{(k)})-Q(S^{(k-1)})\neq0$ (which is not guaranteed by the condition $\bP_{s^{(k)}}[0\prec S^{(k-1)}]\neq0$ from the procedure), then $V^{(k)}$ can be interpreted as the incremental efficiency in~\eqref{eqn::vvEqualsIncrEff}, with $V^{(k)}=\frac{R(S^{(k)})-R(S^{(k-1)})}{Q(S^{(k)})-Q(S^{(k-1)})}$.
\end{enumerate}
\end{definition}

\begin{proposition}\label{prop::decreasingVV}
Every product $j\in N$ falls under one of the following two categories:
\begin{enumerate}
\item $j=s^{(k)}$ for some iteration $k=1,\ldots,K$, in which case for all $k'>k$,
$
V^{(k')}\le V^{(k)};
$
\item $j\notin S^{(K)}$, in which case there is some $k$ for which $\bP_j[0\prec S^{(k)}]=0<\bP_j[0\prec S^{(k-1)}]$, with
$0\ge r^{(k)}_j=r^{(k+1)}_j=\cdots=r^{(K)}_j.$
\end{enumerate}
\end{proposition}

The dichotomy identified in Proposition~\ref{prop::decreasingVV} is important for our analysis.
The key ingredient in its proof (see Appendix~\ref{apx:mc}) is that at each iteration $k$, by adding the \textit{most efficient} product $s^{(k)}$ to the assortment $S^{(k-1)}$, the efficiencies of the other products can only go down.
This leads to the first part of Proposition~\ref{prop::decreasingVV}, saying that the virtual valuations defined by our procedure are monotone with $V^{(1)}\ge\cdots\ge V^{(K)}$.
Note that these virtual valuations can eventually be negative.

The second statement in Proposition~\ref{prop::decreasingVV} says that any product which never gets added to the assortment in our procedure must have been ``killed'' at some iteration $k$, where $\bP_j[0\prec S^{(k)}]$ decreased to 0 after $s^{(k)}$ was added.
It says that the externality-adjusted revenue $r^{(k)}_j$ must be non-positive at the end of this iteration $k$, and moreover, never changes again from future updates.
These products $j\in N\setminus S^{(K)}$ which are not added by the end
can be interpreted to have a virtual valuation of $-\infty$, since in~\eqref{eqn::incrEff}, the numerator $r^{(K)}_j\le0$, and the denominator $\bP_j[0\prec S^{(K)}]=0$.

Now we establish some statements about the true revenues $R(S)$ of assortments $S$, which do depend on these values.
These statements, in conjunction with Proposition~\ref{prop::decreasingVV}, will allow us to establish implementability and insurmountability.

\begin{proposition} \label{prop::mcInduction}
For any $k=0,1,\ldots,K$ and $S\subseteq N\setminus S^{(k)}$, we have that
$
R(S^{(k)}\cup S)-R(S^{(k)})=\sum_{j\in S}r^{(k)}_j\bP[j\prec(S^{(k)}\cup S\setminus j)_+].
$
\end{proposition}
Proposition~\ref{prop::mcInduction} can be derived from \cite{desir2019capacity} by interpreting the RHS as the revenue of assortment $S$ on an ``externality-adjusted instance'', in which products $j\in S^{(k)}$ have been made equivalent to the buyer choosing the zero-revenue no-purchase option, and products $j\notin S^{(k)}$ have had their prices replaced by $r^{(k)}_j$.  We provide a self-contained proof in Appendix~\ref{apx:mc}.

\begin{proposition} \label{prop::vvRelMaintained}
For any $k=0,1,\ldots,K$, we have that $R(S^{(k)})=\sum_{l\in\cQ(S^{(k)})}V(l)P(l)$.
\end{proposition}

Proposition~\ref{prop::vvRelMaintained} is proved
%using Proposition~\ref{prop::mcInduction}
in Appendix~\ref{apx:mc}, and says that for any of the assortments $S^{(k)}$, its revenue matches the ``integral'' of the virtual surplus, $\sum_{l\in\cQ(S^{(k)})}V(l)P(l)$, over lists $l$ which make a purchase from $S^{(k)}$.
Having derived Proposition~\ref{prop::vvRelMaintained}, we are now ready to see that implementability holds.

\begin{corollary}[Implementability] \label{cor::mcImpl}
The VVM $V$ defined
constructed
for
list distribution
$P$ by Definition~\ref{def::procedure} is implementable.
\end{corollary}

Corollary~\ref{cor::mcImpl} holds because by Proposition~\ref{prop::decreasingVV}, the virtual valuations defined satisfy $V^{(1)}\ge\cdots\ge V^{(K)}$.
Therefore, for any threshold $w\in\bR$, the set $\{l\in\Omega:V(l)\ge w\}$ takes the form $\{l\in\Omega:\kappa(l)\le k\}$ for some $k=0,1,\ldots,K$.
By construction in Definition~\ref{def::procedure}, assortment $S^{(k)}$ has the desired $\cQ(S^{(k)})$.
Moreover, by Proposition~\ref{prop::vvRelMaintained}, condition~\eqref{eqn::impl} on $R(S^{(k)})$ is satisfied as equality.
%, completing the proof of Corollary~\ref{cor::mcImpl}.

In Corollary~\ref{cor::mcImpl}, the assortments used to satisfy the condition of implementability happen to be nested in that $S^{(0)}\subsetneq\cdots\subsetneq S^{(K)}$.
However, even with this stronger property, implementability does not imply insurmountability, as discussed in Section~\ref{sec::implInsurm}.
Instead, the following theorem establishes insurmountability by inductively adding elements to any assortment $S$ until it becomes one of the assortments $S^{(k)}$, and showing that the value of $\sum_{l\in\cQ(S)}V(l)P(l)-R(S)$ can only go down, as long as the lowest-indexed missing element is added each time.
The proof is then completed by Proposition~\ref{prop::vvRelMaintained}, which says that $\sum_{l\in\cQ(S)}V(l)P(l)-R(S)$ equals 0 when $S$ is one of the assortments $S^{(k)}$.
The full details of the proof are deferred to Appendix~\ref{apx:mc}.

\begin{theorem}[Insurmountability] \label{thm::mcInsurm}
The VVM $V$ constructed for list distribution $P$ by Definition~\ref{def::procedure} is insurmountable; that is, $R(S)\le\sum_{l\in\cQ(S)}V(l)P(l)$ for all assortments $S\subseteq N$.
\end{theorem}

By Corollary~\ref{cor::mcImpl} and Theorem~\ref{thm::mcInsurm}, every Markov Chain choice model has a VVM which is simultaneously implementable and insurmountable, and tractably computable from Definition~\ref{def::procedure}.
Our main result now immediately follows from Lemmas~\ref{lem::LB},~\ref{lem::UB}.

\begin{theorem} \label{thm::mcMR}
Under Markov Chain choice models, the optimal assortment auction is a computationally tractable Myersonian mechanism, whose revenue equals the expected virtual surplus.
\end{theorem}

 \bibliographystyle{splncs04}
 \bibliography{bibliography}

\appendix\clearpage

\section*{APPENDIX}

\section{Deferred Proofs and Calculations from Section~\ref{sec::assortmentAuctions}} \label{apx:secTwo}

\begin{proof}[Proof of Proposition~\ref{prop::taxation}.]
Take any deterministic truthful mechanism defined by allocation functions $j_1,\ldots,j_m$.
Fix a buyer $i\in M$ and true list $\ell_i\in\Omega$.  Let $T_i$ be the ``taxation'' function defined by $$T_i(\vell_{-i})=\bigcup_{\ell'_i\in\Omega}j_i(\ell'_i,\vell_{-i})$$ which takes the union of all products buyer $i$ could potentially receive based on the ``lies'' $\ell'_i$ she could potentially tell.
Through the incentive-compatibility constraints~\eqref{eqn::incentiveCompatible}, we deduce that $\ell_i(j_i(\ell_i,\vell_{-i}))\le\ell_i(j)$ for all $j\in T_i(\vell_{-i})$.
Through the individual-rationality constraints~\eqref{eqn::individuallyRational}, we deduce that
$\ell_i(j_i(\ell_i,\vell_{-i}))\le\ell_i(0)$.
Since it is obvious that $j_i(\ell_i,\vell_{-i})\in T_i(\vell_{-i})$, the above deductions collectively imply that $j_i(\ell_i,\vell_{-i})=\argmin_{j\in T_i(\vell_{-i})_+}\ell_i(j)$, completing the proof.
\end{proof}

\begin{proof}[Proof of Proposition~\ref{prop::myePricing}.]
Consider the mechanism described in Proposition~\ref{prop::myePricing}.

\underline{Truthfulness}: Consider any buyer $i$ in this mechanism, and consider a realization of the reports of the other buyers $\vell_{-i}$.
Let $j'$ be the lowest-priced product whose virtual valuation $\phi_i(r_{j'})$ would allow buyer $i$ to win the auction.
Note that $j'\neq0$ (because the 0 valuation always has a non-positive virtual valuation), and we write $j'=\infty$ if given this $\vell_{-i}$ it is impossible for buyer $i$ to win the auction.
Now, consider the mechanism defined via the taxation principle by $T_i(\vell_{-i})=\{j'\}$ (or $T_i(\vell_{-i})=\emptyset$ if $j'=\infty$).  Repeating this procedure for all buyers $i$, we have described the mechanism in taxation form, and hence it is truthful.

\underline{Feasibility, and Expected Payment}: We claim that for every realization of $\vell$, the winner and payment is the same as what would have happened had the corresponding valuations been reported in the classical auctions setting.
Indeed, in the classical auctions setting, since the ironed virtual valuations are decreasing, buyer $i$ wins if and only if her valuation is at least $r_{j'}$, and in which case she pays $r_{j'}$.
Meanwhile, in our setting, buyer $i$ will make a purchase if and only if her list realizes to $(1,\ldots,j)$ for some $j\ge j'$, in which case she will choose product $j'$ and pay price $r_{j'}$, since she has buy-down preferences.
Therefore, the mechanism described in Proposition~\ref{prop::myePricing} can have at most one winner, and its expected payment equals exactly the expected virtual surplus, by the known results in the classical auctions setting.

\underline{Optimality}: To see that this mechanism is optimal, we again use the taxation principle.
Consider any feasible, truthful mechanism for assortment auctions under buy-down preferences.
Under buy-down preferences, any non-empty assortment chosen for $T_i(\vell_{-i})$ is equivalent to offering only the lowest-priced product $j'$ in that assortment.  This corresponds to setting a threshold price of $r_{j'}$ for buyer $i$ based on the reports of the other buyers.  By the taxation principle for classical auctions, this translates to a deterministic truthful mechanism, which is feasible because there is at most one winner.
Since in the classical auctions setting, even if we allow for randomized, Bayesian-truthful mechanisms which can charge arbitrary payments, it is not possible to exceed the expected virtual surplus, this completes the proof that it also cannot be exceeded in the assortment auctions setting under buy-down preferences.
\end{proof}

\begin{proof}[Proof of Lemma~\ref{lem::LB}.]
We first apply small perturbations to all the VVM's $V_1,\ldots,V_m$ so that for any realization of $\vell$, the optimization problem of $\max_{M'\in\cF}\sum_{i\in M'}V_i(\ell_i)$ has a unique solution.
This is equivalent to using a fixed deterministic tie-breaking rule when maximizing virtual surplus.
Having done this, we now specify the mechanism by defining $T_i(\vell_{-i})$, the assortment offered to buyer $i$ when the other realizations are $\vell_{-i}$, for all $i$ and $\vell_{-i}$.
This will result in a truthful mechanism, but we must show that it is also feasible, i.e.\ that the set of winners will lie in $\cF$.

For any buyer $i$ and realization of $\vell_{-i}$, there exists a threshold $\tau_i(\vell_{-i})$ such that $i\in\argmax_{M'\in\cF}\sum_{i\in M'}V_i(\ell_i)$ if an only if $V_i(\ell_i)\ge\tau_i(\vell_{-i})$.
Using the condition of implementability, we let $T_i(\vell_{-i})$ be the assortment such that $\cQ(T_i(\vell_{-i}))=\{l\in\Omega:V(l)\ge\tau_i(T_i(\vell_{-i}))\}$.
We claim that for any realization of $\vell$, this results in a feasible allocation.
Indeed, let $M'(\vell)=\argmax_{M'\in\cF}\sum_{i\in M'}V_i(\ell_i)$.
Then
\begin{align*}
i\in M'(\vell)
&\Longleftrightarrow V_i(\ell_i)\ge\tau_i(\vell_{-i}) \\
&\Longleftrightarrow \ell_i\in\{l\in\Omega:V(l)\ge\tau_i(T_i(\vell_{-i}))\} \\
&\Longleftrightarrow \ell_i\in\cQ(T_i(\vell_{-i}))
\end{align*}
and hence the buyers to make a purchase from $T_i(\vell_{-i})$ is exactly $M'(\vell)$, which is feasible.

We now compute the expected revenue of this mechanism, which can be written as $\sum_{i=1}^m\bE_{\vell_{-i}}[R_i(T_i(\vell_{-i}))]$.  By the condition of implementability, this is at least
\begin{align*}
\sum_{i=1}^m\bE_{\vell_{-i}}\Big[\sum_{l\in\cQ(T_i(\vell_{-i}))}V_i(l)P_i(l)\Big]
&=\sum_{i=1}^m\bE_{\vell_{-i}}\Big[\bE_{\ell_i}[V(\ell_i)\bI(\ell_i\in\cQ(T_i(\vell_{-i})))]\Big] \\
&=\bE_{\vell}\left[\sum_{i\in M'(\vell)}V_i(\ell_i)\right]
\end{align*}
where the second equality follows from the tower property of conditional expectation.
This is exactly the expected virtual surplus, completing the proof.
\end{proof}

\begin{proof}[Proof of Lemma~\ref{lem::UB}.]
Fix any feasible, deterministic, truthful mechanism.
By independence, we can separate the total revenue as follows:
\begin{align*}
\bE_{\vell}\left[\sum_{i=1}^mr_{j_i(\vell)}\right]=
\sum_{i=1}^m\bE_{\vell_{-i}}\left[\bE_{\ell_i}\Big[r_{j_i(\ell_i,\vell_{-i})}\Big|\vell_{-i}\Big]\right].
\end{align*}
Now, consider any buyer $i$ and fixed $\vell_{-i}$.  By the taxation principle in Proposition~\ref{prop::taxation}, there must exist an assortment $T_i(\vell_{-i})$ such that for all possible realizations of $\ell_i$, buyer $i$ is allocated her most-preferred item from $T_i(\vell_{-i})_+$.
Therefore, conditioned on any $\vell_{-i}$, we know that $\bE_{\ell_i}[r_{j_i(\ell_i,\vell_{-i})}|\vell_{-i}]=R_i(T_i(\vell_{-i}))$, which is at most $\sum_{l\in\cQ(T_i(\vell_{-i}))}V_i(l)P_i(l)$, by the insurmountability of the virtual valuation mapping $V_i$ for buyer $i$.

Therefore, we can upper-bound the expected revenue by
\begin{align*}
\sum_{i=1}^m\bE_{\vell_{-i}}\left[\sum_{l\in\cQ(T_i(\vell_{-i}))}V_i(l)P_i(l)\right]
&=\sum_{i=1}^m\bE_{\vell_{-i}}\Big[\bE_{\ell_i}[V(\ell_i)\bI(\ell_i\in\cQ(T_i(\vell_{-i})))]\Big] \\
&=\bE_{\vell}\left[\sum_{i=1}^mV_i(\ell_i)\bI(j_i(\vell)\neq0)\right] \\
&\le\bE_{\vell}\left[\max_{M'\in\cF}\sum_{i\in M'}V_i(\ell_i)\right]
\end{align*}
where the second equality follows from the tower property of conditional expectation, and
the inequality holds because for any $\vell$, it must be the case that the set of buyers $i$ for which $j_i(\vell)\neq0$ lies in $\cF$.
This completes the proof.
\end{proof}

\begin{proof}[Proof that Example~\ref{eg::pathological} satisfies claimed properties.]
It can be checked that the revenue frontier is defined by the non-nested sequence of efficient assortments $\emptyset,\{A\},\{A,C\},\{A,B\},\{A,B,C\}$, which have respective sales probabilities $0,\frac{1}{5},\frac{3}{5},\frac{4}{5},1$ and expected revenues 0, 5, 9.8, 12, 13.8.

For this example, trying to construct an implementable VVM $V$ with $\CQC^V(q)=\RF(q)$ for all $q\in[0,1]$ leads to a contradiction.
The crux of the issue is that $\{A,B\}$ follows $\{A,C\}$ in the sequence of efficient assortments, but $\cQ(\{A,B\})=\{(BA),(CB),(B)\}$ is \textit{not a superset of} $\cQ(\{A,C\})=\{(BA),(CB),(C)\}$.
As a result, while normally one would use the efficient assortments to satisfy implementability, in this case the sets $\cQ(S)$ are not nested over $S=\emptyset,\{A\},\{A,C\},\{A,B\},\{A,B,C\}$, which makes it not possible to satisfy the condition $\cQ(S)=\{l\in\Omega:V(l)\ge w\}$ for all $w$.

Using the efficient assortments to define virtual valuations anyway, as described in Section~\ref{sec::introEg}, we would get $V\big((BA)\big)=\frac{5}{1/5}=25$, $V\big((CB)\big)=V\big((C)\big)=\frac{9.8-5}{3/5-1/5}=12$, and $V\big((B)\big)=\frac{12-9.8}{4/5-3/5}=11$.
To see that the resulting Myersonian mechanism may not award a product to the winner, suppose that the buyer's realized list is $(C)$, and that the highest virtual valuation among other buyers is 10.
This buyer would win the auction since her virtual valuation is 12, and she would be allocated her most-preferred product from assortment $\{A,B\}$, whose corresponding slope is 11 (not assortment $\{A,B,C\}$, whose corresponding slope is 9).
Of course, this buyer with list $(C)$ has no item of interest in $\{A,B\}$, hence she must be allocated nothing in order to satisfy individual-rationality, even though she had the highest virtual valuation which was non-negative.
It can be checked that this leads to the virtual surplus not being earned in expectation by the Myersonian mechanism.
\end{proof}

\begin{proof}[Proof that Example~\ref{eg::nonNested} satisfies claimed properties.]
It can be checked that the VVM $V$ defined by $V(BA)=6,V(CB)=2,V(B)=1$ is insurmountable.
The conditions of implementability for $V$ are also satisfied, by setting $S$ to be an assortment from the non-nested sequence $\emptyset,\{A\},\{A,C\},\{B\}$.

The revenue frontier is formed by these assortments, which have respective sales probabilities $0,\frac{1}{3},\frac{2}{3},1$ and expected revenues $0,2,\frac{8}{3},3$.
Moreover, it can be checked that the point $(\frac{2}{3},\frac{8}{3})$ on the revenue frontier can only be attained by assortment $\{A,C\}$, while the point $(1,3)$ can only be attained by assortments $\{B\}$ or $\{A,B\}$.
There is no way to make the latter assortments a superset of the former.
%Indeed, for the threshold $w=2$, the only assortment $S$ which satisfies both $\cQ(S)=\{(BA),(CB)\}$ and $R(S)\ge\frac{6+2}{3}=\frac{8}{3}$ is $S=\{A,C\}$.
%Meanwhile, for the threshold $w=1$, in order for $R(S)\ge\frac{6+2+1}{3}=3$, the assortment $S$ must not contain product~$C$.
Therefore, this distribution has a VVM which is both implementable and insurmountable, even though the revenue frontier cannot be formed by nested assortments.

This introduces an interesting curiosity.
Suppose that a buyer's list, drawn from this distribution, realizes to (CB).
If she is the only buyer in the auction, then she will be offered the ``reserve'' assortment of $\{B\}$, getting her second choice of product~$B$.
However, if there is another buyer such that a virtual valuation of $2$ is required to win the auction, then she will be offered the assortment of $\{A,C\}$, getting her first choice of product~$A$!
This shows that the optimal auction can give a buyer a more-preferred allocation under increased competition, which is something that cannot happen with the optimal auction in the classical Myersonian setting.
\end{proof}

\section{Deferred Proofs from Section~\ref{sec::markovChain}} \label{apx:mc}

First, we state and prove the following basic proposition about transitions on Markov chains, which will be used in our proofs for this section.

\begin{proposition} \label{prop::keyCalculus}
Let $A$ and $B$ be disjoint subsets of $N_+$ and let $c\in N_+\setminus(A\cup B)$.  Then
\begin{align*}
\bP[c\prec A]=\bP[c\prec A\cup B]+\sum_{b\in B}\bP[b\prec c\cup A\cup B\setminus b]\bP_b[c\prec A].
\end{align*}
Furthermore, if $d\neq c$ is another node in $N_+\setminus(A\cup B)$, then
\begin{align*}
\bP_d[c\prec A]=\bP_d[c\prec A\cup B]+\sum_{b\in B}\bP_d[b\prec c\cup A\cup B\setminus b]\bP_b[c\prec A].
\end{align*}
\end{proposition}

\begin{proof}[Proof of Proposition~\ref{prop::keyCalculus}.]
Consider the event $c\prec A$ that $c$ precedes any node in $A$.  For it to occur, the first node visited among \textit{all} the nodes in $c\cup A\cup B$ must be either $c$, or some node $b\in B$.
Conditioned on the first such node visited being $c$, the event $c\prec A$ is guaranteed to occur.
Meanwhile, conditioned on the first such node visited being some $b\in B$, the probability that event $c\prec A$ occurs is equal to the probability of it occurring had the Markov Chain been fixed to start at node $b$, by memorylessness.
Therefore, using the law of total expectation, we get
\begin{align*}
\bP[c\prec A]=\bP[c\prec A\cup B]\cdot1+\sum_{b\in B}\bP[b\prec c\cup A\cup B\setminus b]\cdot\bP_b[c\prec A]
\end{align*}
which is exactly the first statement.  For the second statement, note that fixing the Markov Chain to start at a node $d$ outside of $c\cup A\cup B$ is equivalent to changing the starting probabilities of the Markov Chain so that $\lambda_d=1$, completing the proof.
\end{proof}

\subsection*{Proofs from Section~\ref{sec::markovChain}}

\begin{proof}[Proof of Proposition~\ref{prop::decreasingVV}.]
Consider any iteration $k$, and product $j\neq s^{(k)}$ satisfying $j\in N\setminus S^{(k-1)}$ and $\bP_j[0\prec S^{(k-1)}]\neq0$ which was eligible to be added, but not added, to the assortment during iteration $k$.
When externality-adjusted price $r^{(k)}_j$ is set during iteration $k$, we have
\begin{align}
r^{(k)}_j
&=r^{(k-1)}_j-r^{(k-1)}_{s^{(k)}}\bP_j[s^{(k)}\prec S^{(k-1)}_+] \nonumber \\
&\le r^{(k-1)}_j-\frac{r^{(k-1)}_j}{\bP_j[0\prec S^{(k-1)}]}\bP_j[s^{(k)}\prec S^{(k-1)}_+]\bP_{s^{(k)}}[0\prec S^{(k-1)}] \nonumber \\
&=\frac{r^{(k-1)}_j(\bP_j[0\prec S^{(k-1)}]-\bP_j[s^{(k)}\prec S^{(k-1)}_+]\bP_{s^{(k)}}[0\prec S^{(k-1)}])}{\bP_j[0\prec S^{(k-1)}]} \nonumber \\
&=r^{(k-1)}_j\frac{\bP_j[0\prec S^{(k)}]}{\bP_j[0\prec S^{(k-1)}]} \label{eqn::decrVVRatio}
\end{align}
where the inequality holds because $\frac{r^{(k-1)}_{s^{(k)}}}{\bP_{s^{(k)}}[0\prec S^{(k-1)}]}\ge\frac{r^{(k-1)}_j}{\bP_j[0\prec S^{(k-1)}]}$, and the final equality holds by applying the second part of Proposition~\ref{prop::keyCalculus} with $d=j$, $A=S^{(k-1)}$, $B=\{s^{(k)}\}$, and $c=0$.

There are two cases.  If $\bP_j[0\prec S^{(k)}]=0$, then the inequality ending with~\eqref{eqn::decrVVRatio} implies that $r^{(k)}_j\le0$.
Furthermore, consider any future iteration $k'>k$ where the externality-adjusted price $r^{(k')}_j$ of product $j$ is updated.
We would like to argue that $r^{(k')}_j=r^{(k'-1)}_j$, by arguing that the probability $\bP_j[s^{(k')}\prec S^{(k'-1)}_+]$ used in the update step is zero.
Since product $s^{(k')}$ was added to the assortment during iteration $k'$, it must be the case that $\bP_{s^{(k')}}[0\prec S^{(k'-1)}]>0$.  However, we know that
\begin{align*}
\bP_j[0\prec S^{(k)}]
&\ge\bP_j[s^{(k')}\prec 0\prec S^{(k'-1)}] \\
&=\bP_j[s^{(k')}\prec S^{(k'-1)}_+]\bP_{s^{(k')}}[0\prec S^{(k'-1)}]
\end{align*}
where the inequality holds because we are only adding constraints (note that $S^{(k'-1)}\supseteq S^{(k)}$ since $k'>k$).
Since the LHS equals 0 and the second term on the RHS is strictly positive, it follows that the first term on the RHS must be zero
(because a probability cannot be negative).
This completes the argument that $0\ge r^{(k)}_j=r^{(k+1)}_j=\cdots$, in the first case where $\bP_j[0\prec S^{(k)}]=0$.

In the other case where $\bP_j[0\prec S^{(k)}]>0$, the inequality ending with~\eqref{eqn::decrVVRatio} implies that
$
\frac{r^{(k)}_j}{\bP_j[0\prec S^{(k)}]}\le\frac{r^{(k-1)}_j}{\bP_j[0\prec S^{(k-1)}]}.
$
That is, the incremental efficiency of a product $j$ decreases over the iterations where it is eligible to be added to the assortment.

Putting the two cases together, we arrive at the following understanding of the procedure.
Since $S^{(0)}=\emptyset$, on iteration $k=1$, all products $j\in N$ satisfy $\bP_j[0\prec S^{(k-1)}]\neq0$ and are eligible to be added to the assortment.
Every iteration, a product $s^{(k)}$ gets added, and for all other products $j$ that were eligible to be added during iteration $k$ (these products satisfy $\bP_j[0\prec S^{(k-1)}]\neq0$), either $j$ is made ineligible by the addition of $s^{(k)}$ (this is the second statement of Proposition~\ref{prop::decreasingVV} where $\bP_j[0\prec S^{(k)}]=0$), or the incremental efficiency of $j$ decreases.
If $j$ is eventually added into the assortment at a later iteration $k'>k$ (this is the first statement of Proposition~\ref{prop::decreasingVV}), then
\begin{align*}
\frac{r^{(k'-1)}_{s^{(k')}}}{\bP_{s^{(k')}}[0\prec S^{(k'-1)}]}\le\frac{r^{(k-1)}_{s^{(k')}}}{\bP_{s^{(k')}}[0\prec S^{(k-1)}]}\le\frac{r^{(k-1)}_{s^{(k)}}}{\bP_{s^{(k)}}[0\prec S^{(k-1)}]}
\end{align*}
where the second inequality holds because product $s^{(k)}$ had the highest incremental efficiency during iteration $k$.
This completes the proof.
\end{proof}

\begin{proof}[Proof of Proposition~\ref{prop::mcInduction}.]
Proceed by induction.  If $k=0$, then $S^{(0)}=\emptyset$ and $r^{(0)}_j=r_j$ for all $j$.  Both the LHS and the RHS of the statement describe $R(S)$, the revenue under the original prices from offering assortment $S$.

If $k>0$, then suppose that the statement holds with $k$ replaced by $k-1$.  We can write
\begin{align*}
R(S^{(k)}\cup S)-R(S^{(k)})=(R(S^{(k-1)}\cup s^{(k)}\cup S)-R(S^{(k-1)}))-(R(S^{(k-1)}\cup s^{(k)})-R(S^{(k-1)})).
\end{align*}
By the induction hypothesis, this equals
\begin{align*}
&r^{(k-1)}_{s^{(k)}}\bP[s^{(k)}\prec(S^{(k-1)}\cup S)_+]+\sum_{j\in S}r^{(k-1)}_j\bP[j\prec(S^{(k)}\cup S\setminus j)_+]-r^{(k-1)}_{s^{(k)}}\bP[s^{(k)}\prec S^{(k-1)}_+] \\
= &\sum_{j\in S}r^{(k-1)}_j\bP[j\prec(S^{(k)}\cup S\setminus j)_+]-r^{(k-1)}_{s^{(k)}}(\bP[s^{(k)}\prec S^{(k-1)}_+]-\bP[s^{(k)}\prec(S^{(k-1)}\cup S)_+]) \\
= &\sum_{j\in S}r^{(k-1)}_j\bP[j\prec(S^{(k)}\cup S\setminus j)_+]-r^{(k-1)}_{s^{(k)}}\sum_{j\in S}\bP[j\prec(S^{(k)}\cup S\setminus j)_+]\bP_j[s^{(k)}\prec S^{(k-1)}_+] \\
= &\sum_{j\in S}r^{(k)}_j\bP[j\prec(S^{(k)}\cup S\setminus j)_+].
\end{align*}
where the second equality applies the first statement of Proposition~\ref{prop::keyCalculus} with $A=S^{(k-1)}_+$, $B=S$, and $c=s^{(k)}$, and the final equality follows from the way $r^{(k)}_j$ is defined for all $j\in S\subseteq N\setminus S^{(k)}$.
This completes the induction and the proof.
\end{proof}

\begin{proof}[Proof of Proposition~\ref{prop::vvRelMaintained}.]
Proceed by induction.  When $k=0$, the statement clearly holds.  When $k>0$, suppose that $R(S^{(k-1)})=\sum_{l\in\cQ(S^{(k-1)})}V(l)P(l)$.  By Proposition~\ref{prop::mcInduction},
\begin{align*}
R(S^{(k-1)}\cup s^{(k)})-R(S^{(k-1)})
&=r^{(k-1)}_{s^{(k)}}\bP[s^{(k)}\prec S^{(k-1)}_+] \\
&=\left(\frac{r^{(k-1)}_{s^{(k)}}}{\bP_{s^{(k)}}[0\prec S^{(k-1)}]}\right)\left(\bP[s^{(k)}\prec S^{(k-1)}_+]\bP_{s^{(k)}}[0\prec S^{(k-1)}]\right).
\end{align*}
The first term equals the virtual valuation $V(l)$ of lists $l\in\cL(s^{(k)}\prec0\prec S^{(k-1)})$.
The second term equals the measure of such lists, denoted by $\bP[s^{(k)}\prec0\prec S^{(k-1)}]$.

Therefore, $R(S^{(k-1)}\cup s^{(k)})-R(S^{(k-1)})=\sum_{l\in\cL(s^{(k)}\prec0\prec S^{(k-1)})}V(l)P(l)$.
Combining this with the induction hypothesis, we get that
$
R(S^{(k)})
=\sum_{l\in\cQ(S^{(k-1)})}V(l)P(l)+\sum_{l\in\cL(s^{(k)}\prec0\prec S^{(k-1)})}V(l)P(l).
$
The lists in $\cQ(S^{(k-1)})$, which have $j\prec0$ for some $j\in S^{(k-1)}$, are disjoint from the lists in $\cL(s^{(k)}\prec0\prec S^{(k-1)})$, and their union equals $\cQ(S^{(k)})$.
This completes the induction.
\end{proof}

\begin{proof}[Proof of Theorem~\ref{thm::mcInsurm}.]
Our goal is to prove that the expression
\begin{align} \label{eqn::mcVVcovers}
\sum_{l\in\cQ(S)}V(l)P(l)-R(S)
\end{align}
is non-negative for all $S\subseteq N$.
First we show that this is true if $S$ contains all of the products in $S^{(K)}=\{s^{(1)},\ldots,s^{(K)}\}$.
If not, we show that by adding the lowest-indexed product $s^{(k)}$ currently not in $S$ into $S$, the value of~\eqref{eqn::mcVVcovers} can only decrease.
These two statements combined yield the result.

To see that~\eqref{eqn::mcVVcovers} is non-negative if $S\supseteq S^{(K)}$, let $S'=S\setminus S^{(K)}$, so that $S^{(K)}\cup S'=S$.
Invoking Proposition~\ref{prop::mcInduction} with $k=K$, we derive that
$
R(S)-R(S^{(K)})=R(S^{(K)}\cup S')-R(S^{(K)})=\sum_{j\in S'}r^{(K)}_j\bP[j\prec(S^{(K)}\cup S'\setminus j)_+].
$
Now, by Proposition~\ref{prop::decreasingVV}, $r^{(K)}_j\le0$ for all $j\notin S^{(K)}$, which includes all $j\in S'$.
Therefore, $R(S)$ is at most $R(S^{(K)})$, which equals $\sum_{l\in\cQ(S^{(K)})}V(l)P(l)$ by Proposition~\ref{prop::vvRelMaintained}.
Now, by the stopping criterion, there cannot be any non-empty lists $l\notin\cQ(S^{(K)})$ with $P(l)>0$, and hence $\sum_{l\in\cQ(S^{(K)})}V(l)P(l)=\sum_{l\in\cQ(S)}V(l)P(l)$, completing the proof that~\eqref{eqn::mcVVcovers} is non-negative if $S\supseteq S^{(K)}$.

Now suppose that $s^{(k)}\notin S$, where $k\in\{1,\ldots,K\}$ is the \textit{smallest} index of such a product, i.e.\ $\{s^{(1)},\ldots,s^{(k-1)}\}=S^{(k-1)}\subseteq S$.  Let $S'=S\setminus S^{(k-1)}$, so that $S^{(k-1)}\cup S'=S$.  Invoking Proposition~\ref{prop::mcInduction} for iteration $k-1$, the following can be derived:
\begin{align}
&R(S\cup s^{(k)})-R(S) \nonumber \\
&=(R(S^{(k-1)}\cup S'\cup s^{(k)})-R(S^{(k-1)}))-(R(S^{(k-1)}\cup S')-R(S^{(k-1)})) \nonumber \\
&=\sum_{j\in S'\cup s^{(k)}}r^{(k-1)}_j\bP[j\prec(S^{(k-1)}\cup S'\cup s^{(k)}\setminus j)_+]-\sum_{j\in S'}r^{(k-1)}_j\bP[j\prec(S^{(k-1)}\cup S'\setminus j)_+] \nonumber \\
&=r^{(k-1)}_{s^{(k)}}\bP[s^{(k)}\prec S_+]-\sum_{j\in S'}r^{(k-1)}_j(\bP[j\prec(S\setminus j)_+]-\bP[j\prec(S\cup s^{(k)}\setminus j)_+]) \nonumber \\
&=r^{(k-1)}_{s^{(k)}}\bP[s^{(k)}\prec S_+]-\sum_{j\in S'}r^{(k-1)}_j\bP[s^{(k)}\prec S_+]\bP_{s^{(k)}}[j\prec(S\setminus j)_+] \nonumber \\
&\ge\bP[s^{(k)}\prec S_+]\left(r^{(k-1)}_{s^{(k)}}-\sum_{j\in S':\bP_j[0\prec S^{(k-1)}]\neq0}\left(\frac{r^{(k-1)}_j}{\bP_j[0\prec S^{(k-1)}]}\cdot\bP_j[0\prec S^{(k-1)}]\right)\bP_{s^{(k)}}[j\prec(S\setminus j)_+]\right) \nonumber \\
&\ge\bP[s^{(k)}\prec S_+]\left(r^{(k-1)}_{s^{(k)}}-\sum_{j\in S':\bP_j[0\prec S^{(k-1)}]\neq0}\frac{r^{(k-1)}_{s^{(k)}}}{\bP_{s^{(k)}}[0\prec S^{(k-1)}]}\bP_{s^{(k)}}[j\prec(S\setminus j)_+]\bP_j[0\prec S^{(k-1)}]\right). \label{eqn::2345}
\end{align}
The final equality applies the first statement of Proposition~\ref{prop::keyCalculus} with $A=(S\setminus j)_+$, $B=\{s^{(k)}\}$, and $c=j$.
The first inequality holds because for any $j\in S'$ with $\bP_j[0\prec S^{(k-1)}]=0$, $j$ cannot be one of the products added into the assortment before iteration $k$ (since $S'\cap S^{(k-1)}=\emptyset$), and hence $r^{(k-1)}_j\le0$ by Proposition~\ref{prop::decreasingVV}.
The second inequality holds because product $s^{(k)}$ has the maximum incremental efficiency at iteration $k$.

Meanwhile, since $\{s^{(1)},\ldots,s^{(k-1)}\}\subseteq S$, all of the lists in $\cL(s^{(k)}\prec0\prec S)$ had their virtual valuations defined during iteration $k$, equaling $r^{(k-1)}_{s^{(k)}}/\bP_{s^{(k)}}[0\prec S^{(k-1)}]$.  Note that $\cQ(S\cup s^{(k)})\setminus\cQ(S)$ is exactly the set $\cL(s^{(k)}\prec0\prec S)$, which establishes the following:
\begin{align}
&\sum_{l\in\cQ(S\cup s^{(k)})}V(l)P(l)-\sum_{l\in\cQ(S)}V(l)P(l) \nonumber \\
&=\frac{r^{(k-1)}_{s^{(k)}}}{\bP_{s^{(k)}}[0\prec S^{(k-1)}]}\bP[s^{(k)}\prec0\prec S] \nonumber \\
&=\frac{r^{(k-1)}_{s^{(k)}}}{\bP_{s^{(k)}}[0\prec S^{(k-1)}]}\bP[s^{(k)}\prec S_+]\bP_{s^{(k)}}[0\prec S^{(k-1)}\cup S'] \nonumber \\
&=\frac{r^{(k-1)}_{s^{(k)}}}{\bP_{s^{(k)}}[0\prec S^{(k-1)}]}\bP[s^{(k)}\prec S_+]\left(\bP_{s^{(k)}}[0\prec S^{(k-1)}]-\sum_{j\in S'}\bP_{s^{(k)}}[j\prec(S^{(k-1)}\cup S'\setminus j)_+]\bP_j[0\prec S^{(k-1)}]\right) \nonumber \\
&=\bP[s^{(k)}\prec S_+]\left(r^{(k-1)}_{s^{(k)}}-\frac{r^{(k-1)}_{s^{(k)}}}{\bP_{s^{(k)}}[0\prec S^{(k-1)}]}\sum_{j\in S':\bP_j[0\prec S^{(k-1)}]\neq0}\bP_{s^{(k)}}[j\prec(S\setminus j)_+]\bP_j[0\prec S^{(k-1)}]\right). \label{eqn::3456}
\end{align}
The third equality applies the second part of Proposition~\ref{prop::keyCalculus} with $d=s^{(k)}$, $A=S^{(k-1)}$, $B=S'$, and $c=0$.

Now, since expressions~\eqref{eqn::2345} and~\eqref{eqn::3456} are identical, we see that
$$\left(\sum_{l\in\cQ(S\cup s^{(k)})}V(l)P(l)-R(S\cup s^{(k)})\right)-\left(\sum_{l\in\cQ(S)}V(l)P(l)-R(S)\right)\le0,$$
i.e.\ the value of expression~\eqref{eqn::mcVVcovers} after adding $s^{(k)}$ into assortment $S$ is no greater than before.
We can iteratively apply this argument to conclude that
\begin{align*}
\sum_{l\in\cQ(S)}V(l)P(l)-R(S)\ge\sum_{l\in\cQ(S\cup S^{(K)})}V(l)P(l)-R(S\cup S^{(K)}),
\end{align*}
and we have already shown that the RHS is non-negative, completing the proof.
\end{proof}

\end{document}